\documentclass[12pt]{article}
\usepackage[margin=1.25in]{geometry}
\usepackage[dvipsnames,usenames]{xcolor}
\usepackage{amsmath,amssymb,amsthm,amsfonts}
\usepackage{mathtools}
\usepackage[T1]{fontenc}
\usepackage{newtxtext,newtxmath}
\usepackage{setspace}
\usepackage{enumitem}
\usepackage[authoryear,round]{natbib}
\usepackage{array}
\usepackage{multirow}
\usepackage{graphicx}
\usepackage[footnotesize,bf]{caption}
\usepackage{url}
\usepackage{lscape}
\usepackage{afterpage}
\usepackage{tikz}
\usetikzlibrary{arrows.meta}
\usepackage[colorlinks,allcolors=blue]{hyperref}
\usepackage[capitalize,nameinlink]{cleveref}

\newtheorem{theorem}{Theorem}

\newtheorem{corollary}{Corollary}
\newtheorem{definition}{Definition}
\newtheorem{example}{Example}
\newtheorem{lemma}{Lemma}
\newtheorem{proposition}{Proposition}
\newtheorem{remark}{Remark}
\newcommand{\Mcal}{\mathcal{M}}
\newcommand{\Scal}{\mathcal{S}}
\newcommand{\Wcal}{\mathcal{W}}
\newcommand{\Ucal}{\mathcal{U}}
\newcommand{\EE}{APME}
\newcommand{\BNE}{BNE}
\newcommand{\NE}{NE}
\newcommand{\DSE}{DSE}
\newcommand{\closure}{\textrm{closure}}
\newcommand{\N}{\mathbb{N}}
\newcommand{\R}{\mathbb{R}}
\newcommand{\supp}{\operatorname{supp}}

\begin{document}
\bibliographystyle{elsart-harv}
\begin{singlespace}
\title{Robust implementation in approximate payoff-monotone equilibrium\footnote{Thanks to Samson Alva, James Andreoni, Olivier Bochet, Antonio Cabrales, Marco Castillo, Yeon-Koo Che, Cary Deck, Cary Fryman, Simon Grant, Huiyi Guo, Yuval Heller, Mark Isaac, Eric Maskin, Salvatore Nunnari, Marek Pycia, and Utku \"{U}nver for useful comments. We also thank audiences at the 25+ Years of QRE conference at Caltech and the 2024 Australia New Zealand Workshop on Experimental Economics, and seminar participants at Boston College, Georgetown/Maryland, NC State University, Ohio State University, UC San Diego, UT Dallas, the University of Zurich, and Universidad Nacional de San Luis. We thank the authors of \citet{Andreoni-Che-Kim-2007-GEB} for sharing their data. An earlier version of this paper circulated as \citet{Velez-Brown-2020-SP-v5}. All errors are our own.}}
\date{\today}

\author{ Rodrigo A. Velez\thanks{Department of Economics, University of Texas at San Antonio, San Antonio, TX 78249
\href{mailto:rodrigo.velez@utsa.edu}{rodrigo.velez@utsa.edu}; \href{https://sites.google.com/site/rodrigoavelezswebpage/home}{https://sites.google.com/site/rodrigoavelezswebpage/home} } \ and  Alexander L. Brown\thanks{Department of
Economics, Texas A\&M University, College Station, TX 77843;
 \href{mailto:alexbrown@tamu.edu}{alexbrown@tamu.edu}; \href{http://people.tamu.edu/\%7Ealexbrown}{http://people.tamu.edu/$\sim$alexbrown}}
 }
\maketitle
\thispagestyle{empty}
\vspace{-0.5cm}
\begin{abstract}
We revisit robust implementation through the lens of empirically plausible behavior. In private-value environments, robustness calls for strategy-proof mechanisms, yet experiments show weakly dominated, near-equilibrium behavior persists. Unlike secure implementation \citep{Saijo-et-al-2007-TE}, which requires that every equilibrium perform, we require it only of approximate payoff-monotone equilibria \citep[\EE,][]{VELEZ-BROWN-2025-GEB}, equilibria consistent with experimental regularities. Our main results are a necessary and a sufficient condition for full implementation in \EE\ at a state; we further identify cases in which the two coincide. The direct mechanism is robust for all common priors if and only if the social choice function is strategy-proof and non-bossy, and for all interior priors if and only if it is strategy-proof and weakly non-bossy. We characterize the information structures where the second-price auction fails.
\medskip

\medskip

\medskip

\end{abstract}
\end{singlespace}

\section{Introduction}\label{Sec:intro}
Robust implementation asks the designer to choose mechanisms that perform well across information environments rather than under a single prior. In private-value models, this pressure leads to strategy-proof social choice functions (SCFs) \citep{Bergemann-Morris-2005-Eca}. Yet a large empirical literature shows that agents often fail to identify dominant strategies and that weakly dominated but near-equilibrium behaviors persist, sometimes sustaining outcomes the designer does not intend.\footnote{See for instance \citet{HEALY-2006-JET,Andreoni-Che-Kim-2007-GEB,Cason-et-al-2006-GEB}.}

Three theoretical responses confront this challenge. Secure implementation, in its robustness interpretation, requires that every Bayesian Nash equilibrium (\BNE) induces the SCF under every common prior information structure \citep{Saijo-et-al-2007-TE}. Full robust implementation extends the requirement beyond common prior information structures \citep{Bergemann-Morris-2011-GEB}. Obvious strategy-proofness (OSP) seeks mechanisms where dominance is transparent enough to be used reliably \citep{Li-AER-17,MacKenzie-2019,Pycia-Troyan-2019}. These approaches offer strong guarantees but are often too restrictive for settings of practical interest.

We take a different route. Like secure implementation, we evaluate mechanisms via their equilibrium behavior across common prior information structures, an approach suited for experienced agents for whom equilibrium play is plausible. However, we do not require that all \BNE\ be desirable, only the empirically plausible ones: we adopt the approximate payoff-monotone equilibrium (\EE) of \citet{VELEZ-BROWN-2025-GEB}, the \BNE\ that can be approximated by frequencies of play in the empirical content of the models that usually rationalize behavior in experiments.\footnote{\EE\ are the \BNE\ that are limits of regular QRE \citep{Goeree-Holt-Palfrey-2005-EE}, or equivalently of stochastic-choice equilibrium models in which players are near-best-responders with small convex costs (Sec.~\ref{Sec:model}).} We then study robust full implementation in \EE: a mechanism should obtain the SCF for every \EE\ across a domain of common prior models.

Our analysis rests on two observations about the direct-revelation game of a strategy-proof SCF. First, in every \EE, each type reports truthfully with positive frequency, because payoff-monotone behavior plays dominant strategies at least as often as any other report. A lie that is strictly worse than truth against some profile of the other agents' types that the liar considers possible is therefore never played: the agent's information exposes it. Second, only lies that change the outcome matter. Whether every \EE\ induces the SCF at a state thus depends on two kinds of properties of the types realized there. One is structural, in the spirit of non-bossiness: a type is \emph{non-bossy} if each of its lies that changes the outcome is costly for it. The other is informational, and we refer to its forms as \emph{restraint}: a type is \emph{restrained at} a profile of the other agents' types if its information exposes each of its lies that changes the outcome against that profile, and it is \emph{restrained} if this holds at every such profile. Neither property implies the other, and a type that has either is harmless. A state is \emph{collectively restrained} if the types' information exposes each joint lie that changes the outcome. This is the condition under which positive weight on dominant strategies alone guarantees the SCF at the state (Lemma~\ref{Lem:collective}).

Our first two results are local: they identify, for each prior and each state, whether every \EE\ induces the SCF at that state. Theorem~\ref{Thm:necessity} is a necessary condition: each type at the state must be restrained at the other agents' true types. If a single type has a lie that changes the outcome against those types and that its information does not expose, then some \EE\ fails at the state. The theorem holds for every dominant-strategy mechanism with an \EE\ that induces the SCF under some interior prior. Its proof builds the failing \EE\ as a limit of behavior generated by the models that are popular for the analysis of experimental data: when the condition fails, such behavior can be close to mutual best responses and still produce undesirable outcomes. Theorem~\ref{Thm:sufficiency} is a sufficient condition based on individual tests of non-bossiness and restraint. If at most one type at the state is bossy and not restrained, and that type passes the test of Theorem~\ref{Thm:necessity}, the state is collectively restrained, so every \EE\ induces the SCF there. In this case the two theorems characterize the states at which the mechanism performs. Otherwise, the conditions leave a gap, which comes from joint misreports: those that are harmless one at a time may change the outcome together.

Three applications follow. Corollary~\ref{Th:str-pr} shows that the direct-revelation mechanism of an SCF fully implements it in \EE\ for every common prior if and only if the SCF is strategy-proof and non-bossy. Thus top trading cycles (TTC) \citep{Shapley-Scarf-1974-JMathE}, median voting \citep{Moulin-1980-PC}, and the uniform rule \citep{Sprumont-1991} are robust in \EE, although none of them is securely implementable: the equilibria that rule out their secure implementation are not empirically plausible. Corollary~\ref{Thm:Interior} shows that for every interior prior, i.e., when each agent considers every profile of the others' types possible, the requirement weakens to strategy-proofness and weak non-bossiness. By \citet{Saijo-et-al-2007-TE}, these are the SCFs that their direct-revelation mechanism fully implements in dominant strategies. So full implementation in dominant strategies is vindicated, for empirically plausible behavior, when no agent rules out any profile of the others' types. What matters is the support of beliefs rather than their dispersion, as we make precise in Sec.~\ref{Sec:interior}. Our last application shows that the gap between the two theorems can be closed when the payoff structure of the SCF is at hand. Proposition~\ref{Prop:SPA} characterizes, state by state, the information structures under which some \EE\ of the second-price auction \citep{Vickrey-1961-JF}, a strategy-proof and bossy SCF, fails to induce it; these are exactly the states that are not collectively restrained. Failures have two sources: a single bidder who can change the price or the tie-breaking lottery at no cost against the realized values, or bidders who share the second-highest value and shade their bids together, each protected by the others as long as it acts alone.

These results have three implications. First, robustness is interim and graded rather than all-or-nothing: a mechanism that fails under some priors can be guaranteed to perform at identifiable states. Second, our results give a qualitative account of a folk wisdom of experimental economics, that dominant-strategy mechanisms perform better under incomplete information,\footnote{For instance, private rather than common information can improve market outcomes \citep[e.g.,][]{smith1994economics}; a typical rationale is that information about others' valuations encourages deviations from one-shot Nash equilibrium to capture extra rents.} as in the second-price auction experiment of \citet{Andreoni-Che-Kim-2007-GEB}, which we discuss in Sec.~\ref{Sec:QRE}: the more profiles agents consider possible, the more lies their information exposes. Refinements that discard weakly dominated equilibria cannot deliver this message.\footnote{\label{Fn:Jackson}For an SCF, \citet{Jackson-1992-ReStud} shows that implementation in undominated strategies by bounded mechanisms requires strategy-proofness. Conversely, a strategy-proof SCF with unique dominant strategies is implemented in undominated strategies by its direct mechanism, whatever the information structure. This approach rules out by assumption the weakly dominated equilibria that can be empirically relevant.} Third, our results are testable: \EE\ place prior-free restrictions on the frequencies of dominant and dominated messages (Lemma~\ref{Lem:interior}), and under our sufficient conditions, payoff-monotone behavior close to equilibrium cannot produce outcomes far from those of the SCF.

Our work draws on the literatures on full implementation, robust mechanism design, social choice, and behavioral mechanism design, which we discuss in Section~\ref{Sec-other-implementation-approaches}. Two features of our approach stand out. First, the solution concept is grounded in the data, and our results use only part of its content. Payoff monotonicity restricts behavior in two ways: messages with higher expected payoffs are played more often, and messages with equal expected payoffs are played equally often. The experimental support, which we review in Sec.~\ref{Sec:model}, concerns the first restriction. Our sufficient conditions use only one of its consequences: in each \EE\ of a dominant-strategy mechanism, each weakly dominant message is played at least as often as under uniform random play (Lemma~\ref{Lem:interior}), and in particular with positive frequency, a property we call baseline dominance \citep[see also][]{Velez-Brown-2020-SP-v5}. Theorem~\ref{Thm:sufficiency}, Corollaries~\ref{Th:str-pr} and~\ref{Thm:Interior}, and Proposition~\ref{Prop:SPA} thus guarantee performance for every baseline-dominant \BNE, and hence for every refinement of \BNE\ between baseline-dominant \BNE\ and \EE: whatever regularities one adds to the description of behavior, such as spite, framing, or rounding, the guarantees survive as long as dominant strategies keep positive weight. Our necessary conditions, in turn, construct each failing equilibrium as a limit of behavior from the noisy best-response models commonly fitted to experimental data, so the failures they identify are sustained by behavior with the full structure of payoff monotonicity. Since we model robustness over common priors, our results are directly comparable with those of \citet{Saijo-et-al-2007-TE} and \citet{Adachi-2014-GEB}. Second, although our methods are those of social choice, our conditions combine the structure of the SCF with the information of the agents: they trace how the information of each type interacts with the properties of the SCF to tighten dominant-strategy incentives. When information is itself a design variable, our results identify the information structures that protect the performance of a mechanism.

More broadly, our results are a first step toward a theory of implementation whose primitive assumptions are only on observables: they identify the mechanisms and information structures for which behavior with the regularities usually observed in experiments guarantees performance whenever it is also close enough to equilibrium. The second step, a quantitative theory of proximity that turns implementation results into testable predictions, is the subject of a companion paper \citep{Velez-Brown-statistical}.

Sec.~\ref{Sec:model} introduces the model. Sec.~\ref{Sec:Results} presents the necessity of strategy-proofness (Sec.~\ref{Sec:SP-necessary}), the \EE\ of two leading examples (Sec.~\ref{TT and 2ndPrice auction}), our necessary and sufficient conditions and their applications (Sec.~\ref{Sec:main}), and the comparison with secure implementation (Sec.~\ref{Sec-Secure-implemen}). Sec.~\ref{Sec:conclusion} relates our approach to the literature and concludes.

\section{Model}\label{Sec:model}

There is a group of agents $N:=\{1,\dots,n\}$ and a finite set of alternatives~$Z$. Agents have private values, i.e., each $i\in N$ has a payoff type~$\theta_i$, determining a \emph{payoff index} $\pi_i(\cdot|\theta_i):Z\rightarrow\R$. To easily accommodate random mechanisms, we write $\pi_i|_{\theta_i}$ for the function $x\in\Delta(Z)\mapsto\pi_i(x|\theta_i):=\mathbb{E}_x[\pi_i(\cdot|\theta_i)]$.\footnote{For a finite set~$F$, $\Delta(F)$ denotes the simplex of probability measures on~$F$. For $q\in\Delta(F)$, $\supp(q):=\{f\in F:q(f)>0\}$ is the support of~$q$.} We identify each $z\in Z$ with the lottery that places probability one on it.  The set of possible payoff types for agent~$i$ is~$\Theta_i$ and the set of possible payoff type profiles is $\Theta:= \prod_{i\in N}\Theta_i$. We assume that~$\Theta$ is finite. For each $S\subseteq N$, $\Theta_S$ is the cartesian product of the type spaces of the agents in~$S$. The generic element of~$\Theta_S$ is~$\theta_S$. When $S=N\setminus\{i\}$ we simply write~$\Theta_{-i}$  and~$\theta_{-i}$. We concatenate partial profiles, as in $(\mu_i,\theta_{-i})$. We use this notation consistently when operating with vectors (as in strategy profiles).

Information is summarized by a common prior $p\in\Delta(\Theta)$. We say that~$p$ is \emph{interior} if it places positive probability in each $\theta\in\Theta$. For each $i\in N$ and each $\theta_i\in\Theta_i$ with $p(\theta_i)>0$, let $p|_{\theta_i}\in\Delta(\Theta_{-i})$ be the distribution of the other agents' types under~$p$ conditional on agent~$i$ drawing type~$\theta_i$. The \emph{conditional support} of type~$\theta_i$ is
\[S_i(p|_{\theta_i}):=\{\tau_{-i}\in\Theta_{-i}:(\theta_i,\tau_{-i})\in\supp(p)\},\]
the set of profiles of the other agents' types that~$\theta_i$ considers possible. If~$p$ is interior, $S_i(p|_{\theta_i})=\Theta_{-i}$; if~$p$ places probability one on~$\theta$, $S_i(p|_{\theta_i})=\{\theta_{-i}\}$.

A finite mechanism is a pair $(M,\varphi)$ where $M:=(M_i)_{i\in N}$ are finite message spaces and $\varphi:M\rightarrow \Delta(Z)$ is an outcome function. Given the common prior~$p$, $(M,\varphi)$ determines a standard Bayesian game $\Gamma:=(M,\varphi,p)$. When the prior is degenerate, i.e., places probability one on a payoff type $\theta\in\Theta$, we refer to this as a game of complete information and denote it simply by $(M,\varphi,\theta)$.  A strategy for agent~$i$ in~$\Gamma$ is a function that assigns to each  $\theta_i\in \Theta_i$  a distribution $\sigma_i(\cdot|\theta_i)\in\Delta(M_i)$. We denote a profile of strategies by $\sigma:=(\sigma_i)_{i\in N}$. For each $S\subseteq N$, and each $\theta_S\in \Theta_S$,  $\sigma_S(\cdot|\theta_S)$  is the corresponding product measure $\prod_{i\in S}\sigma_i(\cdot|\theta_i)$. When $S=N$ we simply write $\sigma(\cdot|\theta)$.  We identify each $m_i\in M_i$ with the strategy that places probability one on it; thus $(m_i,\sigma_{-i})$ is the profile in which~$i$ plays~$m_i$ and the others play~$\sigma_{-i}$. In a game of complete information $(M,\varphi,\theta)$ we write $\sigma_i(\cdot)$ instead of $\sigma_i(\cdot|\theta_i)$. We say that a strategy profile~$\sigma$ is \emph{interior}, and write $\sigma\gg0$, if the support of each $\sigma_i(\cdot|\theta_i)$ is~$M_i$. We endow each $\Delta(M_i)$ with the Euclidean norm $||\cdot||$ and the strategy space $\Sigma(\Gamma):=\Delta(M_1)^{\Theta_1}\times\dots\times\Delta(M_n)^{\Theta_n}$ with the corresponding product topology.

Let $\theta_i\in\Theta_i$ be realized with positive probability under~$p$. The expected payoff of agent~$i$ with type~$\theta_i$ in $\Gamma=(M,\varphi,p)$ from playing strategy~$\mu_i$ when the other agents select actions as prescribed by~$\sigma_{-i}$ is
\[\mathbb{E}_{p|_{\theta_i},(\mu_i,\sigma_{-i})}[\pi_i|_{\theta_i}\circ\varphi]:= \sum \pi_i(\varphi(m)|\theta_i)\,p|_{\theta_i}(\theta_{-i})\,\sigma_{-i}(m_{-i}|\theta_{-i})\,\mu_i(m_i|\theta_i),\]
where the summation is over all $\theta_{-i}\in \Theta_{-i}$ and $m\in M$. For example, $\mathbb{E}_{p|_{\theta_i},(m_i,\sigma_{-i})}[\pi_i|_{\theta_i}\circ\varphi]$ is the expected payoff of type~$\theta_i$ from message~$m_i$ when the others play~$\sigma_{-i}$. A strategy profile~$\sigma$ is a \emph{Bayesian Nash equilibrium} (\BNE) of~$\Gamma$ if, for each $i\in N$ and each~$\theta_i$ realized with positive probability under~$p$, each message in the support of $\sigma_i(\cdot|\theta_i)$ maximizes the expected payoff of type~$\theta_i$ when the others play~$\sigma_{-i}$. We denote the set of \BNE\ of~$\Gamma$ by $\BNE(\Gamma)$. This is an interim notion: it places no restriction on the behavior of types realized with probability zero. Message~$n_i$ \emph{weakly dominates} message~$m_i$ for type~$\theta_i$ in $(M,\varphi)$ if $\pi_i(\varphi(n_i,m_{-i})|\theta_i)\geq\pi_i(\varphi(m_i,m_{-i})|\theta_i)$ for each $m_{-i}\in M_{-i}$. A message is \emph{weakly dominant} for~$\theta_i$ if it weakly dominates every message in~$M_i$. Mechanism $(M,\varphi)$ is a \emph{dominant-strategy mechanism} if each type of each agent has a weakly dominant message. A strategy profile~$\sigma$ is a \emph{dominant strategy equilibrium} (\DSE) of $(M,\varphi)$ if, for each $i\in N$ and each $\theta_i\in\Theta_i$, each message in the support of $\sigma_i(\cdot|\theta_i)$ is weakly dominant for~$\theta_i$. We denote the set of \DSE\ of $(M,\varphi)$ by $\DSE(M,\varphi)$. Unlike \BNE, this definition is not interim: it restricts every type, including those with probability zero under the prior, so that \DSE\ is a prior-free prediction, as in \citet{Saijo-et-al-2007-TE}. The choice is immaterial for our purposes: each \DSE\ is a \BNE\ for every prior, and our notions of implementation evaluate outcomes only at type profiles in the support of the prior, where the interim version of \DSE\ and ours coincide.

Our basis for plausibility of behavior is the requirement that, at the interim stage, frequencies of play be ordinally equivalent to expected payoffs.\footnote{\label{Fn:evidence}The evidence is indirect: monotone models fit experimental data well across a large body of studies; see the survey of the quantal response equilibrium literature in \citet{Goeree-Holt-Palfrey-2016-Book} and also \citet{G-L-AER-2021}. This evidence concerns the ranking of messages by expected payoff; our sufficient conditions use only a consequence of it (Lemma~\ref{Lem:interior}).}

\begin{definition}[\citealp{VELEZ-BROWN-2025-GEB}]\rm\label{Def:payoff-monotone} A profile of strategies for $\Gamma:=(M,\varphi,p)$, $\sigma:=(\sigma_i)_{i\in N}$, is \textit{payoff-monotone for~$\Gamma$} if for each $\theta\in\Theta$ in the support of~$p$, each $i\in N$, and each pair $\{m_i,n_i\}\subseteq M_i$,  $\sigma_i(m_i|\theta_i)\geq \sigma_i(n_i|\theta_i)$ if and only if  $\mathbb{E}_{p|_{\theta_i},(m_i,\sigma_{-i})}[\pi_i|_{\theta_i}\circ\varphi]\geq \mathbb{E}_{p|_{\theta_i},(n_i,\sigma_{-i})}[\pi_i|_{\theta_i}\circ\varphi]$. We say that~$\sigma$ is \textit{weakly payoff-monotone for~$\Gamma$} when the if direction $(\Leftarrow)$ in these statements is satisfied. Let $\Mcal(\Gamma)$ and $\Wcal(\Gamma)$ be the set of payoff-monotone and weakly payoff-monotone strategies for~$\Gamma$, respectively.\end{definition}

A \BNE\ is plausible if observing behavior near it does not necessarily reject payoff monotonicity.

\begin{definition}[\citealp{VELEZ-BROWN-2025-GEB}]\rm\label{Def:APME} A \textit{(Bayesian) approximate payoff-monotone equilibrium (\EE)} of $\Gamma:=(M,\varphi,p)$ is a \BNE\ of~$\Gamma$ that is the limit of a sequence of payoff-monotone strategies for~$\Gamma$. Equivalently,
\[\EE(\Gamma)=\BNE(\Gamma)\cap \closure(\Mcal(\Gamma)).\]
\end{definition}

For each~$\Gamma$, $\EE(\Gamma)$ coincides with the limits, as the perturbation parameters vanish, of equilibria in perturbed equilibrium models in which the agents' payoffs are of the form
\begin{equation}\mathbb{E}_{p|_{\theta_i},(\mu_i,\sigma_{-i})}[\pi_i|_{\theta_i}\circ\varphi]-\sum_{m_i\in M_i}c_i(\mu_i(m_i|\theta_i)|\theta_i),\label{EQ:choice-eq-models}\end{equation}
for a smooth strictly decreasing and strictly convex $c_i(\cdot|\theta_i)$ with an asymptote at zero \citep{VELEZ-BROWN-2025-GEB}. Since these perturbed equilibria always exist and their limits, as perturbations vanish, are \BNE\ \citep{VanDamme-1991-Springer},  for each~$\Gamma$, $\EE(\Gamma)$ is non-empty.

\EE\ does not discard all weakly dominated \BNE\ (see the second-price auction in Sec.~\ref{TT and 2ndPrice auction}), nor does it admit all undominated equilibria: it discriminates among trembling-hand perfect and even proper equilibria.\footnote{In a game with multiple weakly dominant strategies, all dominant strategy equilibria are proper, but the only dominant strategy \EE\ is the one in which agents uniformly randomize among their weakly dominant strategies. See \citet{VELEZ-BROWN-2025-GEB} for the comparison with other refinements.} The refinement is also sensitive to the introduction of redundant strategies, a feature we exploit in Sec.~\ref{Sec:revelationpple}.

\citet{VELEZ-BROWN-2025-GEB} show that for each~$\Gamma$,
\begin{equation}
    \closure(\Wcal(\Gamma))=\closure(\Mcal(\Gamma))=\closure(\{\sigma\in\Mcal(\Gamma):\sigma\gg0\}).
\label{Eq:VB-equiv}\end{equation}

It follows that \EE\ can always be approached by interior payoff-monotone behavior and that \EE\ are the \BNE\ that can be approached by weakly payoff-monotone behavior.

A (random) social choice function (SCF) selects a distribution of outcomes for each possible state.\footnote{We allow SCFs to be stochastic to account for tie-breakers, which are common in applications.} The generic SCF is $g:\Theta\rightarrow \Delta(Z)$. Given an SCF~$g$, a strategy profile~$\sigma$ of $(M,\varphi,p)$ \emph{induces~$g$ at $\theta\in\supp(p)$} if $\varphi(m)=g(\theta)$ for each~$m$ in the support of $\sigma(\cdot|\theta)$.\footnote{When~$g$ is random, this is an equality of lotteries: each message profile that is played with positive probability at~$\theta$ must produce the lottery~$g(\theta)$. It is not enough that the realized outcome be in the support of~$g(\theta)$, or that the mixture of $\varphi(m)$ over~$m$ drawn from $\sigma(\cdot|\theta)$ equal~$g(\theta)$.} It \emph{induces~$g$ for prior~$p$} if it induces~$g$ at each $\theta\in\supp(p)$, i.e., if $\Pr_{p,\sigma}[\varphi(m)=g(\theta)]=1$, where~$\theta$ is drawn from~$p$ and~$m$ from $\sigma(\cdot|\theta)$.

\BNE, \EE, and \DSE\ are solutions for the games induced by a mechanism, i.e., they can be seen as functions that determine a set of distributions for each $(M,\varphi,p)$.
\begin{definition}\rm \label{Definition-Implementation}Let $\Scal$ be a solution and~$g$ an SCF. Then, $(M,\varphi)$ \emph{fully implements~$g$ in solution $\Scal$ for prior~$p$} if:
\begin{enumerate}
    \item[(i)] there is $\sigma\in \Scal(M,\varphi,p)$, such that~$\sigma$ induces~$g$ for~$p$.

    \item[(ii)]  for each $\sigma\in \Scal(M,\varphi,p)$, $\sigma$ induces~$g$ for~$p$.
\end{enumerate}
\end{definition}
Since $\DSE(M,\varphi)$ does not depend on the prior, $(M,\varphi)$ fully implements~$g$ in \DSE\ for every prior if and only if it does so for some interior prior. In this case we simply say that $(M,\varphi)$ fully implements~$g$ in \DSE.

Three properties of SCFs play an important role in our results. An SCF~$g$,

\begin{enumerate}
  \item is \textit{strategy-proof} if for each $\theta\in\Theta$, each $i\in N$, and each $\tau_i\in\Theta_i$, $\pi_i(g(\theta)|\theta_i)\geq \pi_i(g(\tau_{i},\theta_{-i})|\theta_i)$.
  \item is \textit{non-bossy} if for each $\theta\in\Theta$, each $i\in N$, and each $\tau_i\in\Theta_i$, if $g(\theta)\neq g(\tau_{i},\theta_{-i})$, then $\pi_i(g(\theta)|\theta_i)\neq \pi_i(g(\tau_{i},\theta_{-i})|\theta_i)$.
  \item is \textit{weakly non-bossy} if for each $\theta \in \Theta$, each $i \in N$, and each $\tau_i \in \Theta_i$, if $g(\theta) \neq g(\tau_i,\theta_{-i})$ and $\pi_i(g(\theta)|\theta_i)=\pi_i(g(\tau_{i},\theta_{-i})|\theta_i)$, then there is $\tau_{-i} \in \Theta_{-i}$ such that
  $\pi_i(g(\theta_i,\tau_{-i})|\theta_i) \neq \pi_i(g(\tau_i,\tau_{-i})|\theta_i)$.
\end{enumerate}

The first property is well-known. Non-bossiness requires that a deviation from a truthful report that alters the outcome at a type profile~$\theta$ have welfare consequences for the deviating agent at~$\theta$. It is the noncorruptibility of \citet{Ritz-1985-JET}, a welfare version of the non-bossiness of \citet{Satterthwaite-Sonn-1981-RES}, under which a report that does not change the agent's own assignment does not change the outcome. \citet{Saijo-et-al-2007-TE} introduced it in implementation, together with weak non-bossiness, and we follow their terminology. Weak non-bossiness requires the welfare consequences only at some type profile in which the agent has the same type. For strategy-proof SCFs, it is equivalent to requiring that every dominant strategy of an agent in the direct-revelation game yield the same outcome, i.e., that no agent has a \emph{bossy dominant strategy}. This holds whenever each type has a unique weakly dominant action, as in the second-price auction, but also, for instance, in the SPDA mechanism for school choice \citep{Abdulkadiroglu-Sonmez-2003-AER}, where a student with strict preferences may have several dominant strategies (e.g., when ranked first by all schools), none of which changes the outcome \citep[see][]{Velez-Brown-2020-SP-v5}.

\section{Results}\label{Sec:Results}

\subsection{The necessity of strategy-proofness}\label{Sec:SP-necessary}
The following lemma suggests that robust implementation in \EE\ concerns strategy-proof SCFs.

\begin{lemma}\rm\label{Lm:full-interior-implementation-necessary}
Let~$g$ be an SCF. Suppose that there is $(M,\varphi)$ such that for each interior prior~$p$ there is an \EE\ of $(M,\varphi,p)$ that induces~$g$ for~$p$. Then~$g$ is strategy-proof.
\end{lemma}

Lemma~\ref{Lm:full-interior-implementation-necessary} is closely related to \citet[Proposition 3]{Bergemann-Morris-2005-Eca}, who obtain the same conclusion from pure-strategy equilibria. Since some games have no \EE\ in pure strategies, our result does not follow from theirs, but their argument applies, and we present the proof in the appendix for completeness.

We thus study when the direct-revelation mechanisms of strategy-proof SCFs are robust under the \EE\ prediction, and how they fail when they are not.

\subsection{\EE\ of TTC and second-price auction}\label{TT and 2ndPrice auction}
Two stylized examples illustrate how \EE\ discriminates between two popular dominant-strategy mechanisms: the top trading cycles (TTC) mechanism for the reallocation of indivisible goods from individual endowments and the second-price auction.

\begin{example}[TTC]\rm\label{Ex:TTC}
Consider two agents, $\{A,B\}$, with strict preferences over the houses they each own $\{h_A,h_B\}$. Under TTC, each agent points to a house. They swap houses if each agent points to the other agent's house; they remain in their houses otherwise. It is a dominant strategy for each agent to point to their preferred house, so dominant-strategy play produces trade whenever both agents prefer to trade. Other equilibria exist: if each agent unconditionally points to their own house, these strategies are mutual best responses under any information structure, and thus a \BNE, but they produce no trade.
\end{example}

Let us analyze the \EE\ of the TTC.

\begin{remark}\rm Each \EE\ of the TTC is equivalent to the dominant strategy equilibrium.
\end{remark}

To see why, consider first the TTC game in a complete information environment in which both agents prefer to trade. The normal form of this game is shown in Figure~\ref{EE-graph1} (left), where the payoff of remaining in one's own house is zero. Strategies are represented by 
$\sigma_i(h_j)$, the probability with which agent~$i$ points to agent~$j$'s house, where $(\sigma_{A}(h_B),\sigma_{B}(h_A))\in[0,1]\times[0,1]$.
There are two \NE\ in this game.\footnote{Following convention, we refer to \BNE\ in complete information environments as Nash equilibria (\NE).} Agents either coordinate, $(\sigma_A(h_B),\sigma_B(h_A))=(1,1)$; or they miscoordinate, $(\sigma_A(h_B),\sigma_B(h_A))=(0,0)$. Only the former is an \EE. This result is apparent from the closure of the set of payoff-monotone distributions (the shaded area in Figure~\ref{EE-graph1} (right)). Since TTC with strict preferences is non-bossy, Corollary~\ref{Th:str-pr} below extends this conclusion to each common prior.

\begin{figure}[t]
\centering
\begin{minipage}[c]{0.5\textwidth}\centering
{\renewcommand{\arraystretch}{1.5}%
\newcolumntype{P}{>{\centering\arraybackslash$}p{1.1cm}<{$}}%
\begin{tabular}{r c|P|P|}
  \multicolumn{2}{c}{} & \multicolumn{2}{c}{Agent~$B$}\\[2pt]
  \multicolumn{2}{c}{} & \multicolumn{1}{c}{$h_A$} & \multicolumn{1}{c}{$h_B$}\\\cline{3-4}
  \multirow{2}{*}{Agent~$A$} & $h_A$ & 0,\ 0 & 0,\ 0\\\cline{3-4}
  & $h_B$ & a,\ b & 0,\ 0\\\cline{3-4}
\end{tabular}}
\end{minipage}%
\begin{minipage}[c]{0.4\textwidth}\centering
\begin{tikzpicture}[scale=2.5,font=\small]
  \fill[gray!30] (0.5,0.5) rectangle (1,1);
  \draw (0,0) rectangle (1,1);
  \draw[dashed] (0.5,0) -- (0.5,1) (0,0.5) -- (1,0.5);
  \node[below] at (0,0) {$0$};
  \node[below] at (0.5,0) {$1/2$};
  \node[below] at (1,0) {$1$};
  \node[left] at (0,0.5) {$1/2$};
  \node[left] at (0,1) {$1$};
  \node[below=14pt] at (0.5,0) {$\sigma_A(h_B)$};
  \node[rotate=90,above=24pt] at (0,0.5) {$\sigma_B(h_A)$};
  \fill (1,1) circle (0.6pt) node[above right=-1pt] {\EE};
  \draw[fill=white] (0,0) circle (0.6pt);
  \node[above right=-1pt] at (0,0) {\NE};
\end{tikzpicture}
\end{minipage}
\caption{\label{EE-graph1} (left) TTC game when each agent prefers to trade: $a:=\pi_A(h_B|\theta_A)>0$ and $b:=\pi_B(h_A|\theta_B)>0$. (right) Strategy-profile space. The shaded area is the closure of payoff-monotone behavior. There are two \NE, $(0,0)$ and $(1,1)$. Only $(1,1)$, in which agents trade houses, is an \EE.
}
\end{figure}

\begin{example}[Second-price auction]\rm\label{Ex:SPA}
The second-price auction \citep{Vickrey-1961-JF} allocates a seller's good to one of $n\geq2$ bidders.\footnote{The term ``second-price auction'' denotes the SCF defined by the dominant-strategy outcomes of the mechanism in which agents bid for the good and the winner pays the second highest bid.} Each bidder~$i$ has a value~$\theta_i$ for the good on the grid $V:=\{1,\dots,K\}$. Bidders have quasi-linear preferences that assign zero payoff when receiving no object and $\theta_i-x_i$ when receiving the object and paying~$x_i$. Each bidder reports a value in~$V$. A bidder with the highest report receives the object and pays the highest report among the other bidders. Ties are decided uniformly at random, so a bidder with value~$v$ who wins at price~$b$ in a tie with $k-1$ other bidders obtains $(v-b)/k$. Truthfully revealing values is weakly dominant for each bidder under this mechanism. The truthful dominant-strategy equilibrium attains an efficient assignment of the object, i.e., a bidder with the highest value receives the object. Moreover, the revenue of the seller is the second highest value.

Other equilibria exist. Let $K\geq3$, and suppose that bidder~$1$ has value~$K$, bidders $2,\dots,n$ have value $K-1$, and information is complete. In the truthful equilibrium, bidder~1 wins and pays $K-1$. Consider the profile in which bidder~1 bids~$K$ and each other bidder randomizes uniformly on $\{1,\dots,K-1\}$. Bidder~1 plays its weakly dominant bid. Each bidder $i\geq2$ loses with each bid in $\{1,\dots,K-1\}$, as it does with truth, while bid~$K$ ties with bidder~1 at a price above its value. So these bids are best replies, and the profile is a \NE. In this equilibrium, bidder~1 wins and pays the highest bid of the other bidders, which is below $K-1$ with probability $((K-2)/(K-1))^{n-1}$. The seller then obtains lower revenue than in the truthful equilibrium.
\end{example}

\begin{figure}[t]
\centering
\begin{tikzpicture}[font=\small,
  bar/.style={fill=gray!45,draw=black,line width=0.4pt},
  guide/.style={densely dotted,gray},
  flow/.style={-{Latex[length=2.2mm]},thick}]
  \newcommand{\panelaxes}[2]{
    \draw[-{Latex[length=1.6mm]}] (0,#1) -- (0,#2);
    \foreach \x/\l in {0.45/1,1.0/2,1.55/3} \node[below] at (\x,#1) {$\l$};}
  \newcommand{\spbar}[2]{\draw[bar] (#1-0.17,0) rectangle (#1+0.17,#2);}
  \newcommand{\guideline}[2]{\draw[guide] (0,#1) -- (2,#1) node[pos=0,left,black] {#2};}
  \begin{scope}[shift={(0,0)}]
    \panelaxes{0}{2.4}\draw (0,0) -- (2,0);
    \guideline{2}{$1$}\guideline{0.2}{$\varepsilon$}
    \spbar{0.45}{0.2}\spbar{1.0}{0.2}\spbar{1.55}{1.6}
    \node[above] at (1,2.4) {$\sigma_1(\cdot)$};
  \end{scope}
  \begin{scope}[shift={(4.0,1.2)}]
    \panelaxes{-1.2}{1.2}\draw (0,0) -- (2,0);
    \node[left] at (0,0) {$0$};
    \guideline{-1.2}{$-1/2+2\varepsilon$}
    \spbar{0.45}{0.2}\spbar{1.0}{0.4}\spbar{1.55}{-1.2}
    \node[above] at (1,1.2) {Bidder~2's expected payoff};
  \end{scope}
  \begin{scope}[shift={(7.6,0)}]
    \panelaxes{0}{2.4}\draw (0,0) -- (2,0);
    \guideline{2}{$1/2$}\guideline{0.4}{$\varepsilon$}
    \spbar{0.45}{1.8}\spbar{1.0}{1.8}\spbar{1.55}{0.4}
    \node[above] at (1,2.4) {$\sigma_2(\cdot)$};
  \end{scope}
  \begin{scope}[shift={(11.2,0)}]
    \panelaxes{0}{2.4}\draw (0,0) -- (2,0);
    \guideline{2.25}{$3/2$}
    \spbar{0.45}{0.675}\spbar{1.0}{1.6875}\spbar{1.55}{2.025}
    \node[above] at (1,2.4) {Bidder~1's expected payoff};
  \end{scope}
  \draw[flow] (2.2,1.2) -- (2.9,1.2);
  \draw[flow] (6.2,1.2) -- (7.0,1.2);
  \draw[flow] (9.8,1.2) -- (10.6,1.2);
  \draw[flow,rounded corners=6pt] (13.4,1.2) -- (13.85,1.2) -- (13.85,-0.85) -- (-0.65,-0.85) -- (-0.65,1.2) -- (-0.1,1.2);
\end{tikzpicture}
\caption{\label{Fig:EE-SP} Weakly payoff-monotone profile in the second-price auction of Example~\ref{Ex:SPA} with $n=2$ and $K=3$: bidder~1 has value~$3$, bidder~2 has value~$2$, and information is complete; drawn for $\varepsilon=1/10$. From left to right: $\sigma_1$; the expected payoff it gives bidder~2 from each bid; $\sigma_2$, which ranks bids weakly as these payoffs do; and the expected payoff that~$\sigma_2$ gives bidder~1, which ranks bids weakly as~$\sigma_1$ does, closing the cycle. As $\varepsilon$ vanishes, the profile converges to the \NE\ in which bidder~2 randomizes uniformly between bids~$1$ and~$2$.}
\end{figure}

Let us analyze the \EE\ of the second-price auction. The \NE\ in Example~\ref{Ex:SPA} is an \EE, so for some information structures the second-price auction has an \EE\ that raises less revenue than the second highest valuation.\footnote{This is a failure of implementation for a designer who represents a seller expecting the second highest valuation as payment. There is no inefficient \EE\ of the second-price auction \citep{VELEZ-BROWN-2025-GEB}; other \EE\ raise the price or change the tie-breaking lottery. Proposition~\ref{Prop:SPA} characterizes all failures to induce the SCF.} For $\varepsilon>0$ small, let bidder~1 bid~$K$ with probability $1-(K-1)\varepsilon$ and each other bid with probability~$\varepsilon$, and let each bidder $i\geq2$ bid~$K$ with probability~$\varepsilon$ and each bid in $\{1,\dots,K-1\}$ with probability $(1-\varepsilon)/(K-1)$. This profile is weakly payoff-monotone. Since the other bidders' play is interior, bidder~1 faces each profile of rival bids with positive probability, so its unique weakly dominant bid~$K$ yields a strictly higher expected payoff than every other bid. Bid~$K$ has the highest frequency, and all other bids have the same frequency. A bidder $i\geq2$ never pays more than its value with a bid in $\{1,\dots,K-1\}$, so these bids yield a non-negative expected payoff. Bid~$K$ ties with bidder~1 at price~$K$ with probability close to one, so it yields a negative expected payoff. Thus bidder~$i$'s expected payoffs rank bid~$K$ below every other bid, as its frequencies do, and its frequencies place no other restriction, since all bids in $\{1,\dots,K-1\}$ have the same frequency. Figure~\ref{Fig:EE-SP} illustrates for $n=2$ and $K=3$. As $\varepsilon$ vanishes, this profile converges to the \NE\ of Example~\ref{Ex:SPA}, which is thus an \EE\ by (\ref{Eq:VB-equiv}).

\subsection{Full implementation in \EE}\label{Sec:main}
\subsubsection{A necessary condition}\label{Sec:necessary}
Every \EE\ of TTC is truthful equivalent, whereas the second-price auction admits undesirable \EE\ under some information structures. This prompts two questions. First, which structural features of an SCF are responsible for TTC-style robustness? Second, when robustness fails, is the failure categorical or graded: does the second-price auction retain some robustness at identifiable states and priors?

We answer both with two results localized at a state, a necessary and a sufficient condition for every \EE\ to induce the SCF there, and three applications.

We start with a lemma, proved in the appendix, that holds for each dominant-strategy mechanism.

\begin{definition}\rm Let $(M,\varphi)$ be a dominant-strategy mechanism and~$p$ a prior. A strategy profile~$\sigma$ of $(M,\varphi,p)$ is \emph{baseline-dominant} if it prescribes that each type of each agent that has positive probability under~$p$ plays each of its weakly dominant messages with positive frequency. If~$g$ is strategy-proof, truth is weakly dominant in $(\Theta,g)$, so in a baseline-dominant profile of $(\Theta,g,p)$ each such type reports truthfully with positive frequency.
\end{definition}

\begin{lemma}\rm\label{Lem:interior}Let~$p$ be a prior and $\sigma\in \EE(M,\varphi,p)$, where $(M,\varphi)$ is a dominant-strategy mechanism. Then each type~$\eta_i$ with $p(\eta_i)>0$ plays each of its weakly dominant messages with frequency at least $1/|M_i|$; in particular, $\sigma$ is baseline-dominant. Moreover, if message~$n_i$ weakly dominates message~$m_i$ for a type~$\eta_i$ with $p(\eta_i)>0$, then $\sigma_i(n_i|\eta_i)\geq\sigma_i(m_i|\eta_i)$.
\end{lemma}

Both statements of Lemma~\ref{Lem:interior} are testable restrictions that do not involve the prior. In data generated close to an \EE, each weakly dominant message is played at least as often as under uniform random play, and reports are ranked consistently with weak dominance, up to the distance to equilibrium and sampling error. Our sufficient conditions use only the first restriction, and only through baseline dominance.

Our results on bossy SCFs are local: they depend on the types realized at a state. We therefore also apply non-bossiness type by type and, for a type, profile by profile. In the following definition we fix an SCF~$g$.

\begin{definition}\rm\label{Def:local-non-bossiness}
Let $\theta_i\in\Theta_i$. For $\tau_{-i}\in\Theta_{-i}$, individual type~$\theta_i$ is \emph{non-bossy in~$g$ at~$\tau_{-i}$} if for each $\tau_i\in\Theta_i$, if $g(\theta_i,\tau_{-i})\neq g(\tau_{i},\tau_{-i})$, then $\pi_i(g(\theta_i,\tau_{-i})|\theta_i)\neq \pi_i(g(\tau_{i},\tau_{-i})|\theta_i)$. Otherwise, $\theta_i$ is \emph{bossy in~$g$ at~$\tau_{-i}$}. Individual type~$\theta_i$ is \emph{non-bossy in~$g$} if it is non-bossy in~$g$ at each $\tau_{-i}\in\Theta_{-i}$. Otherwise, $\theta_i$ is \emph{bossy in~$g$}.
\end{definition}

Thus~$g$ is non-bossy if and only if each type of each agent is non-bossy in~$g$.

Our results for bossy SCFs depend on which misreports a type's information can rule out. For the following definitions we fix an SCF~$g$ and a common prior~$p$.

\begin{definition}\rm\label{Def:exposed}
Let $S\subseteq N$ and $\theta_S\in\Theta_S$ such that $p(\theta_S)>0$. A joint report~$\tau_S$ is \emph{exposed in~$g$ by the information of~$\theta_S$ in~$p$} if there is $i\in S$ for whom~$\tau_i$ is strictly worse than truth against some profile that~$\theta_i$ considers possible, when some other members of~$S$, with their types in~$\theta_S$, report as in~$\tau_S$ and all remaining agents report truthfully, i.e., there are $i\in K\subseteq S$ and~$\eta_{-K}$ with $(\theta_{K\setminus\{i\}},\eta_{-K})\in S_i(p|_{\theta_i})$ such that $\pi_i(g(\tau_K,\eta_{-K})|\theta_i)<\pi_i(g(\theta_i,\tau_{K\setminus\{i\}},\eta_{-K})|\theta_i)$.
\end{definition}

When $S=\{i\}$, necessarily $K=\{i\}$, so~$\tau_i$ is exposed if it is strictly worse than truth against some truthful profile that~$\theta_i$ considers possible. If~$g$ is strategy-proof, an unexposed lie of an individual type is thus an interim best reply to truthful reports. If~$g$ is also weakly non-bossy, under an interior prior every lie of a type that is not outcome-equivalent to truth is exposed, and under complete information at~$\theta$ a lie of~$\theta_i$ is unexposed if and only if it is costless at~$\theta$.

Exposure identifies joint reports that are never played in an \EE\ of $(\Theta,g,p)$ when~$g$ is strategy-proof: since each type reports truthfully with positive probability (Lemma~\ref{Lem:interior}), if all members of~$S$ played their part of an exposed~$\tau_S$ at~$\theta_S$, the member~$i$ in the definition would face, with positive probability, a profile against which~$\tau_i$ is strictly worse than truth, so~$\tau_i$ would not be a best reply.

The next lemma, whose proof is in the appendix, says that information alone limits which misreports can be played.

\begin{lemma}\rm\label{Lem:unexposed}
Let~$g$ be strategy-proof, $p$ a common prior, and~$\sigma$ a baseline-dominant \BNE\ of $(\Theta,g,p)$. Let $\theta\in\supp(p)$ and~$\tau$ in the support of $\sigma(\cdot|\theta)$. Then~$\tau$ is not exposed in~$g$ by the information of~$\theta$ in~$p$.
\end{lemma}

Thus, when a player's lies are exposed, they are restrained by payoff maximization. For our purposes, we only need to restrain the lies that matter.

\begin{definition}\rm\label{Def:inert}
Let $\theta\in\Theta$ be such that $p(\theta)>0$. Joint report~$\tau$ is a \emph{consequential lie of~$\theta$ in~$g$} if $g(\tau)\neq g(\theta)$; $\theta$ is \emph{collectively restrained in~$g$ by~$p$} if each of its consequential lies in~$g$ is exposed in~$g$ by the information of~$\theta$ in~$p$.
\end{definition}

To guarantee collective restraint, we use individual notions of restraint.

\begin{definition}\rm\label{Def:1-inert}
Let $\theta_i$ be such that $p(\theta_i)>0$. For $\tau_{-i}\in\Theta_{-i}$, individual report~$\tau_i$ is a \emph{consequential lie for~$\theta_i$ at~$\tau_{-i}$} if $g(\tau_i,\tau_{-i})\neq g(\theta_i,\tau_{-i})$; $\theta_i$ is \emph{restrained in~$g$ at~$\tau_{-i}$ by~$p$} if each of its consequential lies at~$\tau_{-i}$ is exposed in~$g$ by its information in~$p$; $\theta_i$ is \emph{restrained in~$g$ by~$p$} if it is restrained in~$g$ at each $\tau_{-i}\in\Theta_{-i}$ by~$p$.
\end{definition}

Restraint of a type is a joint property of the SCF and the type's conditional support, and it is monotone in information: the more profiles a type considers possible, the more consequential lies its information exposes. It is tested against every profile of the other agents' reports, including those the type rules out, while exposure is tested only against the profiles that the type considers possible. Collective restraint of a state~$\theta$ requires that no unexposed joint misreport change the outcome at~$\theta$: it is a group non-bossiness requirement at~$\theta$, localized by information.

The individual notions are related as follows. Let~$g$ be strategy-proof, $p$ a common prior, $\theta\in\supp(p)$, and $i\in N$. Consider four statements about type~$\theta_i$: (N1)~it is non-bossy in~$g$; (N2)~it is non-bossy in~$g$ at~$\theta_{-i}$; (N3)~it is restrained in~$g$ by~$p$; and (N4)~it is restrained in~$g$ at~$\theta_{-i}$ by~$p$. Figure~\ref{Fig:scheme} summarizes the implications between them.

\begin{figure}[t]
\centering
\begin{tikzpicture}[font=\small,
  box/.style={draw,rounded corners=2pt,align=center,inner sep=5pt,text width=3.8cm,minimum height=1.35cm},
  imp/.style={-{Implies},double,double distance=1.8pt,shorten >=3pt,shorten <=3pt},
  no/.style={midway,sloped,font=\normalsize,inner sep=0pt}]
  \node[box] (N1) at (0,0) {(N1) $\theta_i$ is non-bossy\\ in~$g$};
  \node[box] (N3) at (8.6,0) {(N3) $\theta_i$ is restrained\\ in~$g$ by~$p$};
  \node[box] (N2) at (0,-3.4) {(N2) $\theta_i$ is non-bossy\\ in~$g$ at~$\theta_{-i}$};
  \node[box] (N4) at (8.6,-3.4) {(N4) $\theta_i$ is restrained\\ in~$g$ at~$\theta_{-i}$ by~$p$};
  \draw[imp] ([yshift=4pt]N1.east) -- node[no]{$/$} ([yshift=4pt]N3.west);
  \draw[imp] ([yshift=-4pt]N3.west) -- node[no]{$/$} ([yshift=-4pt]N1.east);
  \draw[imp] ([xshift=-6pt]N1.south) -- ([xshift=-6pt]N2.north);
  \draw[imp] ([xshift=6pt]N2.north) -- node[no]{$/$} ([xshift=6pt]N1.south);
  \draw[imp] ([xshift=-6pt]N3.south) -- ([xshift=-6pt]N4.north);
  \draw[imp] ([xshift=6pt]N4.north) -- node[no]{$/$} ([xshift=6pt]N3.south);
  \draw[imp] ([yshift=4pt]N2.east) -- ([yshift=4pt]N4.west);
  \draw[imp] ([yshift=-4pt]N4.west) -- node[no]{$/$} ([yshift=-4pt]N2.east);
  \draw[imp] ([shift={(3pt,-5pt)}]N3.south west) -- node[no]{$/$} ([shift={(3pt,-5pt)}]N2.north east);
  \draw[imp] ([shift={(-3pt,5pt)}]N2.north east) -- node[no]{$/$} ([shift={(-3pt,5pt)}]N3.south west);
\end{tikzpicture}
\caption{\label{Fig:scheme}Non-bossiness and restraint of a type~$\theta_i$ at a state~$\theta\in\supp(p)$, for a strategy-proof SCF~$g$. A double arrow is an implication. A crossed double arrow is an implication that fails for some strategy-proof~$g$, prior~$p$, and state~$\theta$.}
\end{figure}

Two of the implications are immediate: a type that is non-bossy at each profile is non-bossy at~$\theta_{-i}$, and a type that is restrained at each profile is restrained at~$\theta_{-i}$. The third uses strategy-proofness and the fact that~$\theta_i$ considers~$\theta_{-i}$ possible. If~$\theta_i$ is non-bossy at~$\theta_{-i}$ and $g(\tau_i,\theta_{-i})\neq g(\theta)$, then $\pi_i(g(\tau_i,\theta_{-i})|\theta_i)\neq\pi_i(g(\theta)|\theta_i)$, so by strategy-proofness~$\tau_i$ is strictly worse than truth against~$\theta_{-i}$, and it is exposed. More generally, a consequential lie at~$\theta_{-i}$ that is costly at~$\theta_{-i}$ is exposed by~$\theta_{-i}$ itself. So~$\theta_i$ is restrained at~$\theta_{-i}$ if and only if its information exposes each lie that makes it bossy at~$\theta_{-i}$, i.e., each lie that changes the outcome at~$\theta_{-i}$ at no cost.

No other implication holds. Three examples suffice, with the TTC and the second-price auction of Examples~\ref{Ex:TTC} and~\ref{Ex:SPA} (Sec.~\ref{TT and 2ndPrice auction}). First, a non-bossy type need not be restrained. In TTC with complete information at a state at which agent~$A$ prefers to trade and agent~$B$ does not, $A$'s report that ranks its own house first is costless against~$B$'s type, since there is no trade either way. So it is not exposed. It changes the outcome against a report of~$B$ that prefers to trade, which~$A$ rules out. Thus~$A$'s type is non-bossy, since TTC is non-bossy with strict preferences, but it is not restrained. Second, a restrained type may be bossy at~$\theta_{-i}$. Consider the second-price auction with two bidders, an interior prior, and a state at which bidder~2's value is below bidder~1's and above~$1$. Every type is restrained, since the auction is strategy-proof and weakly non-bossy. Yet bidder~2 lowers the price at no cost by lowering its bid. Third, a type that is non-bossy at~$\theta_{-i}$ may be bossy. At the same state, bidder~1 changes the outcome only with a bid that does not exceed bidder~2's value, which is costly. But against a rival with the same value, bidder~1 can drop out of the tie at no cost. The remaining non-implications follow from these three and the implications above. For instance, if (N2) implied (N3), then (N1) would imply (N3).

Thus, non-bossiness and restraint of a type are independent. The key to sufficient conditions for collective restraint is how the two combine when different agents satisfy them. For $\theta\in\supp(p)$, let
\[\Ucal_p(\theta):=\{i\in N:\theta_i\textrm{ is bossy in }g\textrm{ and not restrained in }g\textrm{ by }p\}.\]

We now state our necessary condition.

\begin{theorem}[Necessity]\rm\label{Thm:necessity}
Let~$g$ be strategy-proof and $(M,\varphi)$ a dominant-strategy mechanism such that, for some interior prior~$q$, some \EE\ of $(M,\varphi,q)$ induces~$g$ for~$q$. Let~$p$ be a common prior and $\theta\in\supp(p)$ such that each $\sigma\in\EE(M,\varphi,p)$ induces~$g$ at~$\theta$. Then, for each $i\in N$, type~$\theta_i$ is restrained in~$g$ at~$\theta_{-i}$ by~$p$.
\end{theorem}
A single type is thus enough for failure: if some~$\theta_i$ has a lie that changes the outcome at~$\theta$ and that its information does not expose, some \EE\ fails at~$\theta$, whatever the other agents' types and information.

The proof, in the appendix, constructs the failing \EE. Let type~$\theta_i$ have a consequential lie~$\tau_i$ at~$\theta_{-i}$ that its information does not expose. All other types randomize uniformly on their weakly dominant messages, and~$\theta_i$ randomizes uniformly on its best replies, which include the dominant messages of type~$\tau_i$. This equilibrium is an \EE: it is the limit of weakly payoff-monotone profiles in which the lying type plays logistic quantal responses \citep{Goeree-Holt-Palfrey-2016-Book} and all other types play uniform perturbations of their equilibrium play, with the perturbations vanishing at a suitable speed relative to the precision of the logistic responses.

The hypothesis that some \EE\ of $(M,\varphi)$ induces~$g$ for some interior prior ties the mechanism to~$g$ at every type profile, not only at~$\theta$: since each \EE\ plays every profile of weakly dominant messages with positive probability (Lemma~\ref{Lem:interior}), it holds if and only if $(M,\varphi)$ fully implements~$g$ in \DSE. Some such tie is needed, since the constant mechanism $\varphi\equiv g(\theta)$ is a dominant-strategy mechanism in which every \EE\ induces~$g$ at~$\theta$, whatever~$g$ is. For $(\Theta,g)$, the hypothesis is equivalent to~$g$ being strategy-proof and weakly non-bossy \citep[][see also Corollary~\ref{Thm:Interior} below]{Saijo-et-al-2007-TE}. When weak non-bossiness fails, the failure is total.

\begin{remark}\rm\label{Lem:total-failure}
Let~$g$ be a strategy-proof SCF that violates weak non-bossiness, i.e., some type~$\theta_i$ has a weakly dominant report $\tau_i\neq\theta_i$ with $g(\tau_i,\theta_{-i})\neq g(\theta)$. Then, for every prior~$p$ with $\theta\in\supp(p)$, and in particular for every interior prior, no $\sigma\in\EE(\Theta,g,p)$ induces~$g$ at~$\theta$.
\end{remark}

Payoff-monotone behavior places maximal frequency on each weakly dominant report, and each \EE\ inherits this property (Lemma~\ref{Lem:interior}). Thus, whenever~$\theta$ is realized, $\tau_i$ is played with positive probability and changes the outcome.

This failure is rare in applications: the strategy-proof SCFs of market design, such as TTC with strict preferences, median voting and the uniform rule with single-peaked preferences, and SPDA, have no bossy dominant strategy. It can nevertheless arise, as in the abstract example of \citet{Repullo-1985-RES}, used by \citet{brown2026lieslabelsmechanisms} to study lying aversion experimentally in an adversarial principal-worker environment.

\subsubsection{A sufficient condition}\label{Sec:sufficient}
Collective restraint is the condition under which baseline dominance alone guarantees the outcome of the SCF, and under weak non-bossiness it is the exact condition. The proof of the following lemma is in the appendix.

\begin{lemma}\rm\label{Lem:collective}
Let~$g$ be strategy-proof, $p$ a common prior, and $\theta\in\supp(p)$.
\begin{enumerate}
\item If~$\theta$ is collectively restrained in~$g$ by~$p$, then each baseline-dominant \BNE\ of $(\Theta,g,p)$, and in particular each $\sigma\in\EE(\Theta,g,p)$, induces~$g$ at~$\theta$.\label{Lem:collective-suff}
\item If~$g$ is weakly non-bossy and~$\theta$ is not collectively restrained in~$g$ by~$p$, then some baseline-dominant \BNE\ of $(\Theta,g,p)$ does not induce~$g$ at~$\theta$.\label{Lem:collective-nec}
\end{enumerate}
\end{lemma}

Collective restraint is a condition on joint misreports. Our sufficient condition reduces it to individual tests and requires only strategy-proofness. The reduction rests on the next lemma, which generalizes the unravelling argument of \citet{Saijo-et-al-2007-TE} to incomplete information: in a joint misreport that is not exposed, the misreports of agents whose types are non-bossy or restrained can be undone one at a time without changing the outcome. It does not involve equilibrium, and its proof is in the appendix.

\begin{lemma}\rm\label{Lem:remove}\
Let~$g$ be strategy-proof, $p$ a common prior, and $\theta\in\supp(p)$. If~$\tau$ is a report profile that is not exposed in~$g$ by the information of~$\theta$ in~$p$, then $g(\tau)=g(\tau_{\Ucal_p(\theta)},\theta_{-\Ucal_p(\theta)})$.
\end{lemma}

By Lemma~\ref{Lem:remove}, a failure of collective restraint at~$\theta$ requires a misreporting agent whose type is bossy and not restrained and whose report is an unexposed consequential lie. If $\Ucal_p(\theta)=\{i\}$, the report~$\tau_i$ must change the outcome at~$\theta$ on its own, i.e., against~$\theta_{-i}$. This gives our sufficient condition; the details of the proof are in the appendix.

\begin{theorem}[Sufficiency]\rm\label{Thm:sufficiency}
Let~$g$ be strategy-proof, $p$ a common prior, and $\theta\in\supp(p)$. Suppose that either $\Ucal_p(\theta)=\emptyset$, or $\Ucal_p(\theta)=\{i\}$ for some $i\in N$ and type~$\theta_i$ is restrained in~$g$ at~$\theta_{-i}$ by~$p$. Then~$\theta$ is collectively restrained in~$g$ by~$p$.
\end{theorem}

By Lemma~\ref{Lem:collective}, under the hypotheses of Theorem~\ref{Thm:sufficiency} each baseline-dominant \BNE\ of $(\Theta,g,p)$, and in particular each \EE, induces~$g$ at~$\theta$.

Theorems~\ref{Thm:necessity} and~\ref{Thm:sufficiency} have three messages. First, \emph{bossiness matters only when it is informed}: a bossy type is harmless when its information exposes each of its consequential lies. Second, \emph{robustness is interim}: whether the \EE\ outcome is right at a state depends on the types realized there, on whether they are bossy and what they consider possible, so a mechanism can fail at some states and still be guaranteed to perform at others, where a designer can trust its outcome. Third, \emph{the gap between the two theorems comes from joint misreports}: Theorem~\ref{Thm:necessity} involves single misreports against the true types of the other agents, while Theorem~\ref{Thm:sufficiency} controls joint misreports by requiring that at most one type at~$\theta$ be bossy and not restrained. When this is so, the two theorems characterize the states at which every \EE\ performs. Let~$g$ be strategy-proof and weakly non-bossy, and $|\Ucal_p(\theta)|\leq1$. Then the following are equivalent: each \EE\ of $(\Theta,g,p)$ induces~$g$ at~$\theta$; $\theta$ is collectively restrained in~$g$ by~$p$; and each type~$\theta_j$ is restrained in~$g$ at~$\theta_{-j}$ by~$p$. Indeed, a type that is non-bossy or restrained is restrained at~$\theta_{-j}$ (Fig.~\ref{Fig:scheme}), so the last condition holds automatically for each $j\notin\Ucal_p(\theta)$, and the equivalence follows from Theorems~\ref{Thm:necessity} and~\ref{Thm:sufficiency} and Lemma~\ref{Lem:collective}. When two or more types at~$\theta$ are bossy and not restrained, whether every \EE\ performs can depend on finer details of the payoffs, as we show later in Sec.~\ref{Sec:SPA} and Remark~\ref{Rem:coalitions}. Remark~\ref{Rem:localized} below sharpens the test of Theorem~\ref{Thm:sufficiency} by using only the reports that the types' information does not expose.

\begin{remark}[Beyond collective restraint]\rm\label{Rem:coalitions}
Collective restraint at~$\theta$ is sufficient for each \EE\ to induce~$g$ at~$\theta$ (Lemma~\ref{Lem:collective}), but not necessary. In Example~\ref{Ex:dominance} in the appendix, the SCF is strategy-proof and weakly non-bossy, information is complete at $\theta=(a,a)$, both types are bossy and not restrained, each is restrained at the other's type, and~$\theta$ is not collectively restrained. Yet every \EE\ induces~$g$: one agent's misreport is weakly dominated by a report that makes the partner's misreport costly, and \EE\ frequencies respect weak dominance (Lemma~\ref{Lem:interior}). By Lemma~\ref{Lem:collective}, some baseline-dominant \BNE\ fails at~$\theta$, so here the ranking of reports by weak dominance, and not only positive weight on dominant strategies, rules out the bad equilibrium. The difference from the second-price auction, where collective restraint is also necessary, as we show later in Sec.~\ref{Sec:SPA}, is a payoff detail that Theorems~\ref{Thm:necessity} and~\ref{Thm:sufficiency} cannot see. A sharp characterization at states at which two or more types are bossy and not restrained would have to depend on global properties of the payoff structure.
\end{remark}

\begin{remark}[A localized test]\rm\label{Rem:localized}
The proof of Lemma~\ref{Lem:remove} evaluates~$g$ only at report profiles whose components are not exposed by the information of the corresponding types of~$\theta$. So its test can be weakened. Let $\Ucal^*_p(\theta)$ be the set of agents~$j$ for which there are a report~$\tau_j$ and a profile~$r_{-j}$, with~$\tau_j$ and each component of~$r_{-j}$ not exposed by the information of the corresponding type of~$\theta$, such that $\pi_j(g(\tau_j,r_{-j})|\theta_j)=\pi_j(g(\theta_j,r_{-j})|\theta_j)$ and $g(\tau_j,r_{-j})\neq g(\theta_j,r_{-j})$. A type that is non-bossy or restrained passes this test, so $\Ucal^*_p(\theta)\subseteq\Ucal_p(\theta)$, and a type that passes it is restrained at~$\theta_{-j}$. The proofs of Lemma~\ref{Lem:remove} and Theorem~\ref{Thm:sufficiency} apply with $\Ucal^*_p(\theta)$ in place of $\Ucal_p(\theta)$, and so does the characterization that follows Theorem~\ref{Thm:sufficiency}.
\end{remark}

\subsubsection{Robust implementation for all common priors}\label{Results1}
Our first application answers the first question.

\begin{corollary}[Direct full implementation in \EE\ for all common priors]\rm\label{Th:str-pr}Let~$g$ be an SCF. The following statements are equivalent.
\begin{enumerate}
\item $(\Theta,g)$ fully implements~$g$ in \EE\ for every common prior.\label{Th:str-pr-robustness}
\item $g$ is strategy-proof and non-bossy.\label{Th:str-pr-non-bossy}
\end{enumerate}
\end{corollary}

\begin{proof}
($2\Rightarrow1$) Let~$p$ be a common prior. Every type is non-bossy, so $\Ucal_p(\theta)=\emptyset$ for each $\theta\in\supp(p)$. By Theorem~\ref{Thm:sufficiency}, each $\sigma\in\EE(\Theta,g,p)$ induces~$g$ at each $\theta\in\supp(p)$. Since \EE\ is non-empty, $(\Theta,g)$ fully implements~$g$ in \EE\ for~$p$.

($1\Rightarrow2$) By Lemma~\ref{Lm:full-interior-implementation-necessary}, $g$ is strategy-proof. By statement~1, some \EE\ of $(\Theta,g,q)$ induces~$g$ for an interior prior~$q$. Suppose that~$g$ is bossy: for some~$\theta$, $i$, and~$\tau_i$, $\pi_i(g(\tau_i,\theta_{-i})|\theta_i)=\pi_i(g(\theta)|\theta_i)$ and $g(\tau_i,\theta_{-i})\neq g(\theta)$. Let~$p$ place probability one on~$\theta$. Then~$\tau_i$ is a consequential lie for~$\theta_i$ at~$\theta_{-i}$, and since $S_i(p|_{\theta_i})=\{\theta_{-i}\}$, it is not exposed. So~$\theta_i$ is not restrained at~$\theta_{-i}$. By Theorem~\ref{Thm:necessity}, some \EE\ of $(\Theta,g,p)$ does not induce~$g$ at~$\theta$, which contradicts statement~1.
\end{proof}

Besides TTC with strict preferences, the uniform rule and median voting are strategy-proof and non-bossy (Table~\ref{Tab:Secure-implementation} in Sec.~\ref{Sec-Secure-implemen} below). The second-price auction is bossy and, accordingly, has undesirable \EE\ under some information structures.

\subsubsection{Robust implementation for all interior common priors}\label{Sec:interior}
Our second application concerns interior priors, under which each type considers every profile of the other agents' types possible, so its information exposes each lie that is strictly worse than truth against some profile.

\begin{corollary}[Direct full implementation in \EE\ for all interior common priors]\rm\label{Cor:interior}\label{Thm:Interior}
Let~$g$ be an SCF. The following statements are equivalent.
\begin{enumerate}
\item $(\Theta,g)$ fully implements~$g$ in \EE\ for each \textit{interior} common prior.\label{TH-int-st:implem}
\item $g$ is strategy-proof and weakly non-bossy.\label{TH-int-st:weaknonbossy}
\end{enumerate}
\end{corollary}

\begin{proof}
($2\Rightarrow1$) Let~$p$ be an interior prior, so that $S_i(p|_{\theta_i})=\Theta_{-i}$ for each type~$\theta_i$. Let~$\tau_i$ be a consequential lie for~$\theta_i$ at some $\tau_{-i}\in\Theta_{-i}$. If $\pi_i(g(\tau_i,\tau_{-i})|\theta_i)\neq\pi_i(g(\theta_i,\tau_{-i})|\theta_i)$, let $\eta_{-i}:=\tau_{-i}$. Otherwise, by weak non-bossiness there is $\eta_{-i}\in\Theta_{-i}$ with $\pi_i(g(\tau_i,\eta_{-i})|\theta_i)\neq\pi_i(g(\theta_i,\eta_{-i})|\theta_i)$. By strategy-proofness, $\tau_i$ is strictly worse than truth against~$\eta_{-i}$, which~$\theta_i$ considers possible, so~$\tau_i$ is exposed. Thus every type is restrained, and $\Ucal_p(\theta)=\emptyset$ for each $\theta\in\supp(p)$. By Theorem~\ref{Thm:sufficiency} and the non-emptiness of \EE, $(\Theta,g)$ fully implements~$g$ in \EE\ for~$p$.

($1\Rightarrow2$) By Lemma~\ref{Lm:full-interior-implementation-necessary}, $g$ is strategy-proof. Suppose that~$g$ violates weak non-bossiness: for some~$\theta$, $i$, and~$\tau_i$, $g(\tau_i,\theta_{-i})\neq g(\theta)$ and $\pi_i(g(\tau_i,\eta_{-i})|\theta_i)=\pi_i(g(\theta_i,\eta_{-i})|\theta_i)$ for each $\eta_{-i}\in\Theta_{-i}$. Let~$p$ be an interior prior. Then~$\tau_i$ is a consequential lie for~$\theta_i$ at~$\theta_{-i}$ that is not exposed, so~$\theta_i$ is not restrained at~$\theta_{-i}$. By statement~1 and Theorem~\ref{Thm:necessity}, some \EE\ of $(\Theta,g,p)$ does not induce~$g$ at~$\theta$, which contradicts statement~1.
\end{proof}

Statement~\ref{TH-int-st:weaknonbossy} is equivalent to $(\Theta,g)$ fully implementing~$g$ in \DSE\ (Sec.~\ref{Sec:main}). Since weakly dominated \BNE\ of such SCFs can be empirically relevant (Sec.~\ref{Sec:QRE}), \DSE\ implementation need not a priori guarantee good performance. Corollary~\ref{Thm:Interior} shows that it does on interior information structures. For instance, the second-price auction has undesirable \EE\ under complete information, but no bossy dominant strategy, so every \EE\ induces it for every interior prior.

Corollary~\ref{Thm:Interior} turns on what agents consider possible, not on how dispersed their beliefs are: a profile with arbitrarily small probability is enough to expose a lie. Mix the complete-information prior of Example~\ref{Ex:SPA} with an interior prior, placing probability $\varepsilon>0$ on the latter. The mixture is interior, so every \EE\ induces the auction at each state, although the types at the original state are almost certain of it. For these types, the expected cost of a lie that is costless at that state is of order~$\varepsilon$, so the guarantee may bind only for behavior very close to equilibrium; how close is an open question (Sec.~\ref{Sec:conclusion-final}). The corollary thus captures the effect of enlarging conditional supports, not that of reducing the precision of information.

By contrast, refinements that discard all weakly dominated behavior predict the truthful equilibrium whenever the SCF is weakly non-bossy, regardless of the information structure. By retaining empirically plausible weakly dominated behavior, \EE\ make robustness depend on information, as the data in Sec.~\ref{Sec:QRE} suggest.

\subsubsection{The second-price auction}\label{Sec:SPA}
Many strategy-proof SCFs used in practice are bossy, e.g., the second-price auction and SPDA, so by Corollary~\ref{Th:str-pr} each has an undesirable \EE\ under some prior. We have no general theorem describing exactly where a strategy-proof and bossy SCF fails, but our third application shows that the gap between Theorems~\ref{Thm:necessity} and~\ref{Thm:sufficiency} can be closed in applications: for the second-price auction, we characterize state by state the information structures under which some \EE\ fails.

Let~$g$ be the second-price auction of Example~\ref{Ex:SPA}, with $n\geq2$ bidders and values in $V=\{1,\dots,K\}$. It is strategy-proof and weakly non-bossy, since each type's dominant strategy is unique. For a profile~$\eta_{-i}$ of the other bidders' values, let $h(\eta_{-i}):=\max_{j\neq i}\eta_j$. Against a truthful profile with highest rival value~$h$, a bid $b<v$ is costly for a bidder with value~$v$ exactly when $h\in\{b,\dots,v-1\}$, where it ties instead of winning or loses instead of earning $v-h>0$, and a bid $b>v$ is costly exactly when $h\in\{v+1,\dots,b\}$, where it wins or ties at a price above~$v$. Every type is bossy, since against a rival with the same value a bidder can break the tie at no cost, and each bid $b\neq v$ is a consequential lie at a profile in which a rival has value~$v$. Hence a bidder with value~$v$ is restrained if and only if it considers possible a highest rival value of $v-1$ (if $v>1$) and one of $v+1$ (if $v<K$). By Theorem~\ref{Thm:sufficiency} and Lemma~\ref{Lem:collective}, every \EE\ induces~$g$ at each state at which every bidder considers these two values possible, whatever happens at other states. Such beliefs are far from interior: they can be concentrated near the bidder's own value, correlated, and have many zeros.

At a state~$\theta$, let $v^{(1)}\geq v^{(2)}$ be the two highest values and $T:=\{j\in N:\theta_j=v^{(2)}\}$.

\begin{proposition}\rm\label{Prop:SPA}
Let~$g$ be the second-price auction, $p$ a common prior, and $\theta\in\supp(p)$. The following statements are equivalent.
\begin{enumerate}
\item $\theta$ is collectively restrained in~$g$ by~$p$.\label{Prop:SPA-restrained}
\item Each $\sigma\in\EE(\Theta,g,p)$ induces~$g$ at~$\theta$.\label{Prop:SPA-APME}
\item The following two conditions hold:\label{Prop:SPA-conditions}
\begin{enumerate}[label=(\roman*)]
\item\label{Prop:SPA-local} each type~$\theta_j$ is restrained in~$g$ at~$\theta_{-j}$ by~$p$, and
\item\label{Prop:SPA-collusion} if $v^{(1)}>v^{(2)}>1$ and $|T|\geq2$, then there are $j\in T$ and $\eta_{-j}\in S_j(p|_{\theta_j})$ such that $h(\eta_{-j})\in\{v^{(2)}-1,v^{(2)}\}$ and each rival of bidder~$j$ with value~$v^{(2)}$ at~$\eta_{-j}$ is in~$T$.
\end{enumerate}
\end{enumerate}
\end{proposition}

By Lemma~\ref{Lem:collective}, statement~\ref{Prop:SPA-restrained} is also equivalent to each baseline-dominant \BNE\ of $(\Theta,g,p)$ inducing~$g$ at~$\theta$. Thus, in the second-price auction, the failures that payoff-monotone behavior sustains are exactly those that positive weight on dominant strategies cannot rule out. The SCF of Example~\ref{Ex:dominance} in the appendix behaves differently (Remark~\ref{Rem:coalitions}).

Statement~\ref{Prop:SPA-conditions} is easy to check. Condition~\ref{Prop:SPA-local} is the necessary condition of Theorem~\ref{Thm:necessity}, and only the lies that make a type bossy at~$\theta_{-j}$ need to be exposed (see the discussion of Fig.~\ref{Fig:scheme}). If $v^{(1)}=v^{(2)}$, it holds if and only if each bidder with the highest value is restrained, which exposes its costless bids that change the tie-breaking lottery. If $v^{(1)}>v^{(2)}$, the bidder with the highest value satisfies it, and another bidder~$j$ satisfies it if and only if it considers possible a highest rival value in $\{\theta_j+1,\dots,v^{(2)}+1\}$ when $v^{(1)}>v^{(2)}+1$, which exposes its costless bids that raise the price, and one of $v^{(2)}-1$ when $T=\{j\}$ and $v^{(2)}>1$, which exposes its costless bids that lower the price. Condition~\ref{Prop:SPA-collusion} concerns collusion. When at least two bidders share the second-highest value, the price falls only if all of them shade their bids together. Condition~\ref{Prop:SPA-collusion} fails when each of them is certain not to face a highest rival bid of $v^{(2)}-1$ if the other bidders in~$T$ bid $v^{(2)}-1$ whenever their value is~$v^{(2)}$ and everyone else bids truthfully.

\begin{proof}[Proof of Proposition~\ref{Prop:SPA}]
Statement~\ref{Prop:SPA-restrained} implies statement~\ref{Prop:SPA-APME} by Lemma~\ref{Lem:collective}.

(\ref{Prop:SPA-APME}$\Rightarrow$\ref{Prop:SPA-conditions}) By Corollary~\ref{Thm:Interior} and Theorem~\ref{Thm:necessity}, statement~\ref{Prop:SPA-APME} implies~\ref{Prop:SPA-local}. Assume~\ref{Prop:SPA-local} and suppose that~\ref{Prop:SPA-collusion} fails. We construct an \EE\ that does not induce~$g$ at~$\theta$. By the cost conditions above, for a bidder with value~$v$, each bid $c\leq v$ weakly dominates each bid below~$c$, and bid $v-1$ is worse than bid~$v$ against a highest rival bid of $v-1$ and equivalent to it against every other. Let~$\sigma$ be the profile in which, for each $j\in T$, type~$\theta_j$ bids~$v^{(2)}$ and $v^{(2)}-1$ with probability~$1/2$ each, and every other type bids truthfully. At~$\theta$, with probability $2^{-|T|}$, the price is $v^{(2)}-1$, so~$\sigma$ does not induce~$g$ at~$\theta$. Profile~$\sigma$ is a \BNE: let $j\in T$ and $\eta_{-j}\in S_j(p|_{\theta_j})$, and suppose that the highest rival bid at $(\theta_j,\eta_{-j})$ is $v^{(2)}-1$. Under~$\sigma$, the only bids that differ from a bidder's value are bids of $v^{(2)}-1$ by bidders in~$T$ with value~$v^{(2)}$. Thus $h(\eta_{-j})\in\{v^{(2)}-1,v^{(2)}\}$ and each rival of bidder~$j$ with value~$v^{(2)}$ at~$\eta_{-j}$ is in~$T$, contradicting that~\ref{Prop:SPA-collusion} fails. So the highest rival bid is never $v^{(2)}-1$, and bid $v^{(2)}-1$ is a best reply. For $\varepsilon\in(0,1)$, let~$\sigma^\varepsilon$ play~$\sigma$ with probability $1-\varepsilon$ and randomize uniformly on~$V$ with probability~$\varepsilon$. By \eqref{Eq:VB-equiv}, it is enough to show that~$\sigma^\varepsilon$ is weakly payoff-monotone for small~$\varepsilon$. Since~$\sigma^\varepsilon$ is interior, each type faces each profile of rival bids with positive probability. So a truthful type's unique dominant bid is strictly better in expectation than any other bid, and it has the highest frequency, while all other bids share a frequency. For $j\in T$, type~$\theta_j$ plays $v^{(2)}$ and $v^{(2)}-1$ with the same frequency and all other bids with a smaller common one, so it is enough that $v^{(2)}-1$ is strictly better in expectation than each $c\notin\{v^{(2)}-1,v^{(2)}\}$. For $c<v^{(2)}-1$, this follows from weak dominance, since~$\theta_j$ faces a highest rival bid of $v^{(2)}-1$ with positive probability. Let $c>v^{(2)}$. Bid $v^{(2)}+1$ is a consequential lie for~$\theta_j$ at~$\theta_{-j}$, since it raises the price or ties with~$w$. By~\ref{Prop:SPA-local}, $\theta_j$ considers possible some~$\eta_{-j}$ with $h(\eta_{-j})=v^{(2)}+1$. At $(\theta_j,\eta_{-j})$, the highest rival bid under~$\sigma$ is $v^{(2)}+1$, since no rival bids above its value, and against it~$c$ yields at most $-1/n$ while~$v^{(2)}$ yields zero. So, against~$\sigma$, bid~$c$ is worse than~$v^{(2)}$, and thus than $v^{(2)}-1$, by at least $p|_{\theta_j}(\eta_{-j})/n$, and by continuity it remains worse against~$\sigma^\varepsilon$ for small~$\varepsilon$.

(\ref{Prop:SPA-conditions}$\Rightarrow$\ref{Prop:SPA-restrained}) Assume~\ref{Prop:SPA-local} and~\ref{Prop:SPA-collusion}. Let~$b$ be a report profile with $g(b)\neq g(\theta)$ that is not exposed by the information of~$\theta$. We derive a contradiction. Each~$b_j$ is not exposed by the information of~$\theta_j$ (the case $K=\{j\}$ of Definition~\ref{Def:exposed}), so by~\ref{Prop:SPA-local}, $g(b_j,\theta_{-j})=g(\theta)$. If $v^{(1)}=v^{(2)}$, each bidder with the highest value then bids truthfully and every other bidder bids below~$v^{(1)}$, so $g(b)=g(\theta)$. Thus $v^{(1)}>v^{(2)}$, the bidder~$w$ with the highest value bids above~$v^{(2)}$, and each other bidder bids at most~$v^{(2)}$, exactly~$v^{(2)}$ if it is the only bidder in~$T$. Since $g(b)\neq g(\theta)$, the price $\max_{j\neq w}b_j$ is below~$v^{(2)}$, so $v^{(2)}>1$, $|T|\geq2$, and $b_j<v^{(2)}$ for each $j\in T$. By~\ref{Prop:SPA-collusion}, there are $j\in T$ and $\eta_{-j}\in S_j(p|_{\theta_j})$ such that $h(\eta_{-j})\in\{v^{(2)}-1,v^{(2)}\}$ and each rival of bidder~$j$ with value~$v^{(2)}$ at~$\eta_{-j}$ is in~$T$. Let~$R$ be the set of these rivals, $K:=\{j\}\cup R\subseteq T$, and $i\in K$ a bidder with the lowest bid in~$K$, so that $b_i\leq b_k$ for each $k\in K$ and $b_i<v^{(2)}$. Each member of~$K$ has value~$v^{(2)}$ both at~$\theta$ and at $(\theta_j,\eta_{-j})$. Hence $(\theta_i,\theta_{K\setminus\{i\}},\eta_{-K})=(\theta_j,\eta_{-j})\in\supp(p)$, i.e., $(\theta_{K\setminus\{i\}},\eta_{-K})\in S_i(p|_{\theta_i})$. Consider the profile in which the members of $K\setminus\{i\}$ bid as in~$b$ and the agents outside~$K$ bid their values at~$\eta_{-K}$. No agent outside~$K$ has value~$v^{(2)}$ at~$\eta_{-j}$, since such agents are in $R\subseteq K$, and none has a higher value, since $h(\eta_{-j})\leq v^{(2)}$. So the bids outside~$K$ are at most $v^{(2)}-1$, and the bids in $K\setminus\{i\}$ are below~$v^{(2)}$. Against this profile, the truthful bid~$v^{(2)}$ of bidder~$i$ wins outright at a price~$c$ below~$v^{(2)}$, with payoff $v^{(2)}-c\geq1$. Bid~$b_i$ instead loses, with payoff zero, or ties at price $b_i=c$ with $m\geq2$ bidders, with payoff $(v^{(2)}-b_i)/m<v^{(2)}-c$. In both cases~$b_i$ is strictly worse than truth. Thus~$b$ is exposed by the information of~$\theta$ (Definition~\ref{Def:exposed} with~$i$, $K$, and~$\eta_{-K}$), a contradiction. So each consequential lie of~$\theta$ is exposed, and~$\theta$ is collectively restrained.
\end{proof}

For instance, let $n=3$, $K\geq10$, the support of~$p$ be $\{(10,5,5),(6,5,5)\}$, and $\theta=(10,5,5)$. Condition~\ref{Prop:SPA-local} holds: bidder~1 is non-bossy at $(5,5)$, and a bidder of type~$5$ considers possible a highest rival value of~$6$. Condition~\ref{Prop:SPA-collusion} fails, since a bidder of type~$5$ considers possible only highest rival values of~$10$ and~$6$. In the bad \EE, each bidder of type~$5$ shades its bid to~$4$ with probability~$1/2$, certain that bidder~1 bids at least~$6$ and protected by the other's bid as long as it acts alone. If the support also contains $(1,5,5)$, condition~\ref{Prop:SPA-collusion} holds and every \EE\ induces~$g$ at~$\theta$: shading to~$4$ now ties instead of winning when bidder~1 has value~$1$ and the other bidder of type~$5$ also shades.

The inclusion $\Ucal^*_p(\theta)\subseteq\Ucal_p(\theta)$ of Remark~\ref{Rem:localized} can be strict. In the second-price auction, suppose that condition~\ref{Prop:SPA-local} of Proposition~\ref{Prop:SPA} holds and that $v^{(1)}=v^{(2)}$, $v^{(2)}=1$, or $|T|=1$. By the bid restrictions in the proof of Proposition~\ref{Prop:SPA}, each profile of unexposed reports yields~$g(\theta)$, so $\Ucal^*_p(\theta)=\emptyset$ even if $\Ucal_p(\theta)$ contains several bidders. For instance, let $n=3$, $K=7$, the support of~$p$ be $\{(6,5,3),(6,4,3),(4,5,3),(6,6,3),(6,5,2)\}$, and $\theta=(6,5,3)$. Bidders~1 and~3 are bossy and not restrained, but $\Ucal^*_p(\theta)=\emptyset$. Thus, except when $v^{(1)}>v^{(2)}>1$ and $|T|\geq2$, Proposition~\ref{Prop:SPA} follows from Theorem~\ref{Thm:necessity} and this version of Theorem~\ref{Thm:sufficiency}. In that remaining case, the bidders in~$T$ can lower the price only by acting together, and condition~\ref{Prop:SPA-collusion} is needed.

\subsection{Secure and robust implementation}\label{Sec-Secure-implemen}\label{Sec:secure-implementation}

We now place Corollaries~\ref{Th:str-pr} and~\ref{Cor:interior} in the context of secure implementation \citep{Saijo-et-al-2007-TE}, which requires a dominant-strategy mechanism that fully implements the SCF in \NE\ for all complete information priors. This requirement is equivalent to the existence of a mechanism that fully implements the SCF in \BNE\ for all common prior information structures \citep{Adachi-2014-GEB}. Secure implementation requires that each equilibrium perform; we require it only of \EE.

The following theorem states necessary and sufficient conditions for full implementation in \BNE\ for all common prior information structures in our environment.\footnote{Earlier results restrict to pure-strategy equilibria and consider type spaces larger than our payoff-type space, in which payoff types can be ``cloned,'' so they do not imply the theorem through the purification argument of \citet[Section~6.3]{Bergemann-Morris-2011-GEB}. The equivalence of statements~1 and~4 is Theorem~3 of \citet{Velez-Brown-2020-SP-v5}, which adapts \citet{Saijo-et-al-2007-TE} and \citet{Adachi-2014-GEB}; its proof that~4 implies~1 shows that~4 implies~3. Statement~3 implies~1 trivially, and implies~2 because the direct-revelation mechanism of a strategy-proof SCF is a dominant-strategy mechanism. If every \NE\ of a mechanism induces~$g$, so does every pure-strategy \NE, so~2 implies~4 by \citet{Saijo-et-al-2007-TE}. These results are stated for deterministic SCFs and mechanisms; they apply to our random ones with $\Delta(Z)$ as the outcome set, on which each payoff index is a function and each mechanism is deterministic.}

\begin{theorem}\rm\label{Thm:sec-impl}Let~$g$ be an SCF. The following are equivalent.
\begin{enumerate}
\item (Full implementation in \BNE\ for all common priors) There is a finite mechanism $(M,\varphi)$ that fully implements~$g$ in \BNE\ for all common prior information structures.
\item (Secure implementation) There is a dominant-strategy finite mechanism $(M,\varphi)$ such that for each complete information prior~$\theta$, each \NE\ of $(M,\varphi,\theta)$ induces~$g$.
\item (Direct full implementation in \BNE\ for all common priors) For each common prior~$p$, each equilibrium of $(\Theta,g,p)$ induces~$g$.
\item $g$ is strategy-proof and non-bossy, and satisfies the \emph{outcome-rectangular property}: for each pair $\{\theta,\tau\}\subseteq\Theta$, if $g(\theta_i,\tau_{-i})=g(\tau)$ for each $i\in N$, then $g(\theta)=g(\tau)$.
\end{enumerate}
\end{theorem}

Since each \EE\ is a \BNE, the properties in our characterizations are necessary for these earlier requirements, but not sufficient: full implementation in \BNE\ for all common priors also requires the outcome-rectangular property (statement~4), which is very restrictive in prominent environments that admit non-dictatorial, strategy-proof, and non-bossy SCFs. In single-peaked voting, each securely implementable SCF has a dictator on its range \citep[Theorem~8]{Saijo-et-al-2007-TE}. In Shapley-Scarf housing markets, the efficient SCFs that satisfy it together with strategy-proofness are exactly the sequential dictatorships \citep{Fujinaka-Wakayama-2011-ET}. In allotment economies with single-peaked preferences, the symmetric ones reduce to the equal division rule, and with two agents the efficient ones are exactly the priority rules \citep{Bochet-Sakai-2010-GEB}. Table~\ref{Tab:Secure-implementation} below presents prominent SCFs that satisfy the conditions in our results; most of them violate the outcome-rectangular property and so, by Theorem~\ref{Thm:sec-impl}, cannot be securely implemented by any mechanism. For the strategy-proof and non-bossy SCFs among them, the \BNE\ that rule out secure implementation are not \EE\ under any common prior (Corollary~\ref{Th:str-pr}). For the second-price auction and SPDA, which are bossy, the same holds under every interior prior (Corollary~\ref{Thm:Interior}) and at the states that Theorems~\ref{Thm:necessity} and~\ref{Thm:sufficiency} identify. Secure implementation is a single test of the SCF; when the test fails, it does not say at which priors or states the direct mechanism can be trusted. Our results answer that question for the equilibria that payoff-monotone behavior can approach. The contrast is sharp in house allocation with as many objects as agents. The efficient deterministic SCFs that are strategy-proof and non-bossy are exactly the trading-cycles rules \citep{Pycia-Unver-2017-TE}. By Corollary~\ref{Th:str-pr}, these are the efficient deterministic SCFs whose direct-revelation mechanism fully implements them in \EE\ for all common priors, while secure implementation leaves only the sequential dictatorships.\footnote{The trading-cycles rules that also satisfy reallocation-proofness are the hierarchical exchange rules \citep{Papai-2000-Eca}. The outcome-rectangular property implies reallocation-proofness \citep{Fujinaka-Wakayama-2011-ET}, and neither is needed for robustness in \EE. For these efficient and non-bossy rules, each member of a pair that violates reallocation-proofness does strictly worse than truth against the other's misreport before they swap, so the joint misreport is exposed.} Similarly, in allotment economies with single-peaked preferences, every sequential allotment rule \citep{Barbera-Jackson-Neme-1997-GEB}, including the uniform rule and the priority rules, is robust in \EE.\footnote{These rules are strategy-proof, efficient, and replacement monotonic, and the last property implies the non-bossiness of \citet{Satterthwaite-Sonn-1981-RES}: a report that does not change an agent's own share does not change the allotment \citep{Barbera-Jackson-Neme-1997-GEB}. Since a strategy-proof and efficient rule gives each agent its peak truncated to an interval that depends on the others' preferences \citep{Sprumont-1991,Barbera-Jackson-Neme-1997-GEB}, a costless lie does not change the agent's own share, so the rule is non-bossy in our sense, on every finite set of types.}

Collective restraint of a state (Definition~\ref{Def:inert}) is related to the outcome-rectangular property. Stated with the roles of~$\theta$ and~$\tau$ exchanged, the property says that if $g(\tau_i,\theta_{-i})=g(\theta)$ for each $i\in N$, then $g(\tau)=g(\theta)$: misreports that leave the outcome unchanged one at a time also leave it unchanged together. Thus, under the property, a state is collectively restrained whenever each type at the state is restrained at the other agents' types, and for each strategy-proof and weakly non-bossy~$g$, Theorem~\ref{Thm:necessity} and Lemma~\ref{Lem:collective} characterize the states at which every \EE\ performs. The converse does not hold: collective restraint is imposed only at the realized state and only on joint misreports that the types' information does not expose, and under weak non-bossiness it holds at each state under interior priors (Sec.~\ref{Sec:interior}). The interpretation of the outcome-rectangular property by \citet{Saijo-et-al-2007-TE} starts instead from the misreport profile, typically an equilibrium: if each agent's unilateral return to the truth leaves the outcome unchanged, so does the joint return. The second-price auction also violates the property in this orientation. At the true profile $(10,5,5)$, let both bidders of value~$5$ bid~$6$. Either one's return to the truth leaves bidder~1 winning at price~$6$, but their joint return lowers the price to~$5$.

Our model has a finite type space, while several rules in Table~\ref{Tab:Secure-implementation} are usually defined on infinite domains, so we read each row on a finite set of types. Strategy-proofness, non-bossiness, and the outcome-rectangular property each state that a certain deviation does not exist, so an SCF that satisfies one of them on a domain satisfies it on each finite restriction of the domain. In the rows where non-bossiness holds, it holds in the welfare form that we use: for the allotment rules and generalized medians, the notes show that a lie that changes the outcome is costly, and for TTC and trading cycles, with strict preferences a costless lie leaves the agent's own object unchanged. Weak non-bossiness requires instead that a costless lie that changes the outcome be costly against some other profile of the others' types, which a restriction can remove. So the bossy rows must fix the domain: the second-price auction has values in a common grid $\{1,\dots,K\}$, where each type's dominant strategy is unique (Sec.~\ref{Sec:SPA}), and SPDA has the finite domain of strict preferences.

\begin{table}[!tp]\small
\centering
{\renewcommand{\arraystretch}{1.3}%
\begin{tabular}{l>{\centering\arraybackslash}p{1.75cm}>{\centering\arraybackslash}p{1.75cm}>{\centering\arraybackslash}p{1.75cm}>{\centering\arraybackslash}p{1.75cm}}
\hline
SCF (preference domain) & Strategy-proofness & Weak non-bossiness & Non-bossiness & Outcome-rectangular property\\\hline
Serial dictatorship (strict) & $+$ & $+$ & $+$ & $+$\\
Clarke-Groves (public model (i)) & $+$ & $+$ & $+$ & $+^{a}$\\
Clarke-Groves (public model (ii)) & $+$ & $+$ & $+$ & $+^{a}$\\
TTC (strict) & $+$ & $+$ & $+$ & $-$\\
Trading cycles (strict) & $+$ & $+$ & $+$ & $-^{b}$\\
Sequential allotment (single-peaked) & $+$ & $+$ & $+$ & $-^{c}$\\
Uniform rule (single-peaked) & $+$ & $+$ & $+$ & $-^{d}$\\
Generalized median (single-peaked) & $+$ & $+$ & $+^{e}$ & $-^{f}$\\
Second-price auction (quasi-linear) & $+$ & $+$ & $-$ & $-^{g}$\\
SPDA (strict) & $+$ & $+^{h}$ & $-^{i}$ & $-^{j}$\\\hline
\end{tabular}}

\smallskip
\begin{minipage}{\linewidth}\footnotesize
$^{a}$The mechanisms of \citet{Clarke-1971-PC} and \citet{Groves-1973-Eca} in a quasi-linear domain, for public good cost sharing (model~(i)) and no-cost public good provision with single-peaked preferences (model~(ii)) in \citet[Example~1]{Saijo-et-al-2007-TE}.\\
$^{b}$With as many objects as agents, exactly the sequential dictatorships among these rules satisfy the property \citep{Fujinaka-Wakayama-2011-ET}; TTC does not.\\
$^{c}$\citet{Barbera-Jackson-Neme-1997-GEB}. Among these rules, the priority rules satisfy the property and the uniform rule does not \citep{Bochet-Sakai-2010-GEB}.\\
$^{d}$\citet[Lemma~1 and Example~1]{Bochet-Sakai-2010-GEB}.\\
$^{e}$The generalized median voter schemes of \citet{Moulin-1980-PC}. The outcome is the agent's peak truncated to the interval of outcomes that the agent can obtain by varying its report. So an agent cannot move the outcome toward its peak, or past it, and a lie that changes the outcome moves it away from the peak on the same side and is costly.\\
$^{f}$\citet[Sec.~8]{Saijo-et-al-2007-TE}: only the schemes with a dictator on their range satisfy the property.\\
$^{g}$Sec.~\ref{Sec:SPA}: at $(10,5,5)$, each bidder of value~$5$ alone can bid~$4$ without changing the outcome, but if both do, the price falls. See also \citet[Sec.~7]{Saijo-et-al-2007-TE}.\\
$^{h}$Has strategically equivalent dominant strategies \citep{Velez-Brown-2020-SP-v5}.\\
$^{i}$\citet{Kojima-2010-EL}; see also \citet{Ergin-2002-Eca}.\\
$^{j}$Let two schools with one seat each give top priority to different students, and let each student prefer the school at which the other student has top priority. Truthful reports give each student its preferred school. If both students instead rank first the school at which they have top priority, each gets that school, and this assignment does not change when either student alone reports truthfully.
\end{minipage}
\caption{Strategy-proof SCFs and the outcome-rectangular property; $+$ indicates that the SCF satisfies the property labeling the column, and $-$ that it violates it. Strategy-proofness and non-bossiness are inherited by any finite subdomain of the typical infinite type spaces in which the models are defined.}\label{Tab:Secure-implementation}
\end{table}

\section{Discussion and concluding remarks}\label{Sec:conclusion}

\subsection{Interdependent values and non-common priors}\label{subsec:interdep-nonBayes}

Our analysis focuses on private-value environments with common prior information structures, the setting of a large share of market-design applications, such as school choice, rationing, house allocation, voting, and independent private-values auctions. \citet{Saijo-et-al-2007-TE} were the first to study robust full implementation in this framework, and our scope is aligned with theirs.\footnote{Their results and ours extend to general type spaces \`a la \citet{Bergemann-Morris-2005-Eca} when robust implementation is required only on the common support of the priors; in our state-by-state results, $S_i(p|_{\theta_i})$ is then the support of the belief of type~$\theta_i$. The payoff-type model delivers the policy implications with less notation, and since our theorems are characterizations, stating them on larger prior domains would not strengthen them.} The robust implementation literature has pursued robustness across broader information environments, including non-common priors and interdependent values \citep{Bergemann-Morris-2005-Eca,Bergemann-Morris-2011-GEB}. In private-value models, once non-common priors are admitted, full implementation becomes excessively restrictive, and partial robust implementation reduces to strategy-proofness \citep{Bergemann-Morris-2011-GEB}.

Extending our analysis to interdependent values is an open question. There, robust implementation does not lead to strategy-proofness, which is often overly restrictive, so our methods, which rely on it, do not directly apply.

\subsection{Revelation principle for full implementation in \EE}\label{Sec:revelationpple}

Our analysis focuses on direct-revelation mechanisms. Whether some mechanism fully implements a bossy strategy-proof SCF in \EE\ for all common priors is open, but such a mechanism cannot have dominant strategies. The reason is the next lemma, which extends the necessity of non-bossiness and weak non-bossiness from direct-revelation mechanisms to all dominant-strategy mechanisms. Its proof, in the appendix, applies Theorem~\ref{Thm:necessity}.

\begin{lemma}\rm\label{Lm:full-implem-necessary} Let~$g$ be strategy-proof and $(M,\varphi)$ a dominant-strategy mechanism. Suppose that for each common prior~$p$ in a set~$H$, each $\sigma\in\EE(M,\varphi,p)$ induces~$g$ for~$p$.
\begin{enumerate}
\item If~$H$ is the set of all common priors, then~$g$ is non-bossy.\label{Lm-necessity-non-bossy}
\item If~$H$ is the set of interior common priors, then~$g$ is weakly non-bossy.\label{Lm-necessity-safe}
\end{enumerate}
\end{lemma}

Together with Corollaries~\ref{Th:str-pr} and~\ref{Thm:Interior}, which give sufficiency, and Lemma~\ref{Lm:full-interior-implementation-necessary}, which gives the necessity of strategy-proofness, Lemma~\ref{Lm:full-implem-necessary} yields the following form of the revelation principle; the details are in the appendix.

\begin{theorem}\rm\label{Theorem-revelationpple} Let~$g$ be an SCF.
\begin{enumerate}
\item There is a dominant-strategy mechanism that fully implements~$g$ in \EE\ for each common prior if and only if~$g$ is strategy-proof and non-bossy.
\item There is a dominant-strategy mechanism that fully implements~$g$ in \EE\ for each \textit{interior} common prior if and only if~$g$ is strategy-proof and weakly non-bossy.
\end{enumerate}
\end{theorem}

Thus direct revelation is without loss for robust full implementation in \EE\ within the class of dominant-strategy mechanisms. The obstacle to extending this result to arbitrary mechanisms is that if type~$\theta_i$ has a weakly dominant lie~$\tau_i$ in the direct-revelation mechanism, the equilibrium play of type~$\tau_i$ in an arbitrary mechanism inherits the ``safety'' of the lie for type~$\theta_i$ only relative to the support of the equilibrium: off-equilibrium messages can make it unattractive to~$\theta_i$, so the argument behind Lemma~\ref{Lm:full-implem-necessary} does not extend. This is an obstacle, not a separating example: we know neither of a mechanism without dominant strategies that fully implements a bossy SCF in \EE\ for all common priors, nor of a proof that none exists.

Even within dominant-strategy mechanisms, direct revelation is not without loss for \emph{approximate} implementation, in which, in the spirit of virtual implementation \citep{Abreu-Sen-1991}, a designer tolerates a small failure probability. Indirect mechanisms can then outperform direct revelation through redundant messages.

Let a planner choose an agent's lunch between \emph{healthy} and \emph{unhealthy}, aiming at the healthy option whenever it is weakly best for the agent. The direct-revelation scheme (ask for the preferred option and assign it, breaking ties in favor of healthy) achieves this in each \EE\ with probability one when preferences are strict and $2/3$ when the agent is indifferent. A menu of $k+1$ reports, the first~$k$ mapping to \emph{healthy} and the last to \emph{unhealthy}, achieves it in every \EE\ with probability one under strict preferences and $k/(k+1)$ under indifference. Hence, for any $\varepsilon>0$, some indirect mechanism's \EE\ induce the planner's SCF with probability at least $1-\varepsilon$.\footnote{The exact probabilities in this example, unlike the results of Section~\ref{Sec:Results}, use the restriction that messages with equal expected payoffs are played with equal frequency.}

\subsection{Other implementation approaches}\label{Sec-other-implementation-approaches}
Our closest comparison is with \citet{Saijo-et-al-2007-TE}, who also evaluate mechanisms through their equilibrium behavior, but our work speaks to several strands of the mechanism-design literature motivated by experimental and empirical findings.

A prominent line of research seeks mechanisms with \emph{simple} dominant strategies \citep{Li-AER-17,Pycia-Troyan-2019,MacKenzie-2019}. With few exceptions \citep[e.g.,][]{Arribillaga-et-al-2019}, its simplicity requirements are satisfied essentially only by priority-like SCFs, even in environments that admit more symmetric, non-bossy strategy-proof rules \citep[cf.][]{ASHLAGI-GONCZ-2018-JET,Bade-Gonczarowski-2018,Troyan-2019-IER}. Simplicity aims to promote \emph{dominant strategy} play; we instead evaluate mechanisms by their \emph{equilibrium} behavior among experienced agents, for whom equilibrium play is a reasonable benchmark.

For example, with strict preferences, no undesirable \BNE\ of TTC is an \EE. This is a prediction about the equilibria that payoff-monotone behavior can approach, not about the frequency of truthful reports: payoff-monotone behavior away from equilibrium may well include non-truthful reports, as observed in experiments with TTC, which is not simple in the sense of \citet{Li-AER-17} \citep[see also][]{Troyan-2019-IER}. The two theories are complementary. \citet{Li-AER-17} explains when agents can reliably identify dominant strategies; ours explains when equilibria inconsistent with the designer's objective are unlikely to arise even if dominant-strategy play is imperfect. Where obvious strategy-proofness cannot be achieved but strategy-proofness and non-bossiness can, the designer may focus on interventions that accelerate convergence to equilibrium, such as framing or feedback, without fearing undesirable steady states.

Several recent papers also study the undesirable equilibria of direct mechanisms of strategy-proof SCFs. \citet{Bochet-Tumennasan-2020-JET} show that, for strategy-proof, efficient, and replacement-monotonic rules in the division problem of \citet{Sprumont-1991}, the \NE\ allocations of the direct mechanism form a lattice, ordered by Pareto dominance, whose supremum is the truthful outcome. \citet{Bochet-Sakai-2010-GEB} show that, in the direct mechanism of the uniform rule, each bad \NE\ is blocked by a credible coalitional deviation, while no good one is. \citet{Bochet-Tummenassan-2020} characterize resilience and group resilience, which require that an agent or a group can profitably revert to truthtelling whenever reports lead to an outcome different from the truthful one. An SCF is resilient if and only if it is securely implementable, and group resilient if and only if it is strategy-proof and non-bossy. Thus coalitional incentives also refine away the undesirable equilibria of a strategy-proof and non-bossy SCF. The coincidence of the two classes is no accident: our sufficient conditions use only positive weight on dominant strategies, which payoff-monotone behavior delivers without coalitions, and the class that this positive weight protects at every prior is the strategy-proof and non-bossy one. The state-by-state conditions of Theorems~\ref{Thm:necessity} and~\ref{Thm:sufficiency}, which depend on the information of the realized types, have no counterpart in these papers. \citet{Bochet-Tummenassan-2017} study implementation when agents have benchmark strategies, such as truthtelling, and show that in private-value environments this imposes no restrictions beyond those of standard robust implementation. \citet{Kivinen-Tumennasan-2021-TE} strengthen group strategy-proofness to the interim stage; for instance, acyclic priorities are necessary and sufficient for SPDA to satisfy it. Finally, \citet{Gavan-Penta-2025-TE} require that deviations from equilibrium also lead to acceptable outcomes. These approaches strengthen the requirements on a mechanism or on the coalitions that agents can form. Ours instead keeps the requirement on all equilibria but discards those that payoff-monotone behavior cannot approach.

Our work is also related to \emph{behavioral mechanism design}, which incorporates experimental regularities into the design problem along two main lines. First, \citet{CABRALES2000247}, \citet{HEALY-2006-JET}, and \citet{TUMENNASAN-2013-GEB} study mechanisms through solution concepts defined by convergence processes.\footnote{\citet{TUMENNASAN-2013-GEB} requires that all limits of logistic QRE behavior be optimal, similar in spirit to implementation in \EE, and that at least one such sequence converge strongly, which requires a strict equilibrium and essentially reduces to implementability in pure \NE\ for all complete-information priors \citep{CABRALES2000247}, a drastic narrowing of the class of implementable SCFs (Sec.~\ref{Sec-Secure-implemen}).} These dynamic conditions typically require strict equilibria and are therefore quite restrictive.

Second, another branch analyzes mechanism design under explicit bounded-rationality models \citep[cf.][]{Eliaz-2002-RStud,de-Clippel-2014-AER,de-Clippel-et-al-2017-levelk,Kneeland-2017-levelk}. Our work connects the two: we evaluate mechanisms under a weak but empirically well-supported form of rationality, payoff monotonicity, and require only that behavior be \emph{approximately} in equilibrium, which permits robust predictions without a specific behavioral model or convergence to strict equilibria.\footnote{We do not assume that agents are approximate payoff maximizers. They may \emph{appear} to best-respond because they do or because the behavior of others makes them look as if they did.}
\subsection{Empirical content and experimental data}\label{Sec:QRE}

Equilibrium models of imperfectly rational behavior rationalize experimental data well, but they have rarely been used in mechanism design, because their predictions are fixed points of functional forms that are difficult to characterize even for the simplest parametric families. \EE\ bypasses this difficulty: it aggregates the restrictions shared by these models as perturbations vanish, and our characterizations are stated in terms of the SCF and the information structure. We thus design for agents whose behavior is described by \emph{some} payoff-monotone equilibrium model, drawing on the literature on the empirical content of regular QRE \citep{Goeree-et-al-2018,G-L-AER-2021,VELEZ-BROWN-2025-GEB}. The need to connect this literature with mechanism design was pointed out by \citet{Goeree-Holt-Palfrey-2016-Book}; ours is the first paper to articulate the predictions of these models in economic design.

The conditions that our results identify are consistent with the experimental evidence on strategy-proof mechanisms, particularly that from the second-price auction experiment of \citet{Andreoni-Che-Kim-2007-GEB}. Bidders with independent uniform values play the same value profiles under two information structures: each bidder knows only its own value, an interior prior, or all values. Dominant-strategy play is frequent under both, but far from universal.\footnote{We compute the figures in this paragraph from the authors' data for the second-price auctions, with 1,200 bids and 300 auctions under each information structure, taking the price to be the second-highest bid. They reproduce the authors' average revenues.} A bid equals the bidder's value 68\% of the time under an interior prior and 51\% under complete information.\footnote{The authors, who count a bid within one token of the value as truthful, report 77\% and 58\%, respectively.} Under an interior prior, average revenue is close to the second-highest value, 56.3 tokens against 56.5. Under complete information, losing bidders overbid, in a pattern that the authors attribute to spite, and average revenue rises to 60.7 tokens, 7.4\% above the second-highest value. The price equals the second-highest value in 67\% of the auctions under an interior prior and 37\% under complete information; the average absolute difference between the two is 6.1\% of the average second-highest value in the former and 17.4\% in the latter. These differences fall to 1.0\% and 15.1\%, respectively, in the last ten rounds of the experiment. The allocation stays close to efficient under either structure, with the winner's value on average 98\% of the highest value under an interior prior and 95\% under complete information. This result is consistent with the absence of inefficient \EE\ of the auction \citep{VELEZ-BROWN-2025-GEB}. 

Our results separate the two information structures. Under complete information, a losing bidder's bid above the second-highest value and below the highest value changes the price, and the bidder's information does not expose it. So the bidder's type is not restrained at the other bidders' types, and by Theorem~\ref{Thm:necessity} some \EE\ fails at the state: behavior close to equilibrium that keeps positive weight on dominant strategies can sustain prices above the second-highest value. Under an interior prior, Corollary~\ref{Thm:Interior} rules out every such equilibrium, and by Lemma~\ref{Lem:collective} the guarantee holds for every baseline-dominant \BNE. The guarantee needs only that dominant strategies keep positive weight, whatever drives the deviations from dominant-strategy play. The comparison is qualitative: our results concern conditional supports and do not predict the size of deviations at a given precision of information.

\subsection{Conclusion}\label{Sec:conclusion-final}
We study full implementation based on \EE, the equilibrium refinement that selects the \BNE\ in the closure of payoff-monotone behavior. Robustness of the direct mechanism for all common priors amounts to strategy-proofness and non-bossiness, without the outcome-rectangular property of secure implementation, and for interior priors weak non-bossiness replaces non-bossiness. For bossy SCFs, robustness is interim: it is decided state by state, by whether the information of the types realized there exposes the lies, single or joint, that change the outcome.

Several questions remain open. First, a sharp characterization at states at which two or more types are bossy and not restrained would have to depend on details of the payoff structure (Remark~\ref{Rem:coalitions}), although the gap can be closed in applications (Proposition~\ref{Prop:SPA}). Second, for SPDA, identifying the conditional supports that make student types restrained would extend our analysis of the second-price auction (Sec.~\ref{Sec:SPA}) to school choice. Third, our analysis is restricted to private values (Sec.~\ref{subsec:interdep-nonBayes}) and to full implementation; approximate implementation can benefit from redundant messages (Sec.~\ref{Sec:revelationpple}). Finally, our analysis is silent on how close to equilibrium behavior must be for our guarantees to bind, presumably closer the nearer the information structure is to complete information (Sec.~\ref{Sec:interior}), and on whether payoff-monotone learning dynamics converge to \EE.

\appendix
\section{Proofs}\label{Sec:apx-proofs}
For a mechanism $(M,\varphi)$, a prior~$p$, a type~$\eta_i$ with $p(\eta_i)>0$, a message~$m_i$, and a profile~$\sigma_{-i}$, let
\[U_i(m_i,\sigma_{-i}|\eta_i):=\mathbb{E}_{p|_{\eta_i},(m_i,\sigma_{-i})}[\pi_i|_{\eta_i}\circ\varphi].\]
If~$g$ is strategy-proof, then for each~$i$, $\theta_i$, $\tau_i$, and~$\tau_{-i}$,
\begin{equation}\label{Eq:pointwise}
\pi_i(g(\theta_i,\tau_{-i})|\theta_i)\geq \pi_i(g(\tau_i,\tau_{-i})|\theta_i).
\end{equation}

\begin{proof}[Proof of Lemma~\ref{Lm:full-interior-implementation-necessary}]
Our proof follows \citet[Proposition~3]{Bergemann-Morris-2005-Eca}; we spell out the details because our statement includes mixed-strategy equilibria. Let $\theta\in\Theta$, $i\in N$, and $\tau_i\in\Theta_i$ with $\tau_i\neq\theta_i$. Let $\varepsilon\in(0,1)$, and let~$p$ be the prior that places probability $(1-\varepsilon)/2$ on each of~$\theta$ and $(\tau_i,\theta_{-i})$ and distributes~$\varepsilon$ uniformly on the other type profiles (if there are none, let~$p$ place probability $1/2$ on each of the two). Thus~$p$ is interior. Let $\sigma\in\EE(M,\varphi,p)$ induce~$g$ for~$p$. Then $\varphi(m)=g(\eta)$ for each $\eta\in\Theta$ and each~$m$ in the support of $\sigma(\cdot|\eta)$. Let~$m^{\theta}_i$ and~$m^{\tau}_i$ be in the supports of $\sigma_i(\cdot|\theta_i)$ and $\sigma_i(\cdot|\tau_i)$, respectively. Since~$p$ is interior, for each $\eta_{-i}\in\Theta_{-i}$ and each~$m_{-i}$ in the support of $\sigma_{-i}(\cdot|\eta_{-i})$, $\varphi(m^{\theta}_i,m_{-i})=g(\theta_i,\eta_{-i})$ and $\varphi(m^{\tau}_i,m_{-i})=g(\tau_i,\eta_{-i})$. Since~$m^{\theta}_i$ is a best reply for~$\theta_i$,
\[\sum_{\eta_{-i}\in\Theta_{-i}}p|_{\theta_i}(\eta_{-i})\,\pi_i(g(\theta_i,\eta_{-i})|\theta_i)\geq\sum_{\eta_{-i}\in\Theta_{-i}}p|_{\theta_i}(\eta_{-i})\,\pi_i(g(\tau_i,\eta_{-i})|\theta_i).\]
As $\varepsilon\to0$, $p|_{\theta_i}(\theta_{-i})\to1$, so $\pi_i(g(\theta)|\theta_i)\geq\pi_i(g(\tau_i,\theta_{-i})|\theta_i)$. Thus~$g$ is strategy-proof.
\end{proof}

\begin{proof}[Proof of Lemma~\ref{Lem:interior}]
Let $\sigma\in\EE(M,\varphi,p)$. By \eqref{Eq:VB-equiv}, there is a sequence $\{\sigma^\lambda\}_{\lambda\in\N}$ of weakly payoff-monotone profiles of $(M,\varphi,p)$ converging to~$\sigma$. Let $p(\eta_i)>0$, and suppose that~$n_i$ weakly dominates~$m_i$ for~$\eta_i$. Then the integrand of the expected payoff of~$n_i$ dominates that of~$m_i$ pointwise, so $U_i(n_i,\sigma^\lambda_{-i}|\eta_i)\geq U_i(m_i,\sigma^\lambda_{-i}|\eta_i)$. By weak payoff-monotonicity, $\sigma^\lambda_i(n_i|\eta_i)\geq\sigma^\lambda_i(m_i|\eta_i)$, and the weak inequality survives the limit. This proves the second statement. A weakly dominant message weakly dominates every message, so each weakly dominant message of~$\eta_i$ has the highest frequency, which is at least $1/|M_i|$. In particular, $\sigma$ is baseline-dominant.
\end{proof}

\begin{proof}[Proof of Lemma~\ref{Lem:unexposed}]
We first show that each~$\tau_j$ is costless for~$\theta_j$ against every report profile that~$\theta_j$ faces with positive probability under~$\sigma$. Let $j\in N$. Since~$\tau$ is in the support of the product measure $\sigma(\cdot|\theta)$, $\tau_j$ is in the support of $\sigma_j(\cdot|\theta_j)$. Since~$\sigma$ is a \BNE, $\tau_j$ is a best reply for~$\theta_j$. By \eqref{Eq:pointwise}, so is~$\theta_j$. Hence their expected payoffs are equal. By \eqref{Eq:pointwise}, the integrand of the expected payoff of~$\theta_j$ dominates the integrand of the expected payoff of~$\tau_j$ pointwise, so the two integrands coincide on the support of the integrating measure, which is the set of report profiles that~$\theta_j$ faces with positive probability. Thus $\pi_j(g(\tau_j,r_{-j})|\theta_j)=\pi_j(g(\theta_j,r_{-j})|\theta_j)$ for each report profile~$r_{-j}$ that~$\theta_j$ faces with positive probability under~$\sigma$.

Let $j\in K\subseteq N$ and~$\eta_{-K}$ be such that $(\theta_{K\setminus\{j\}},\eta_{-K})\in S_j(p|_{\theta_j})$. Type~$\theta_j$ faces $(\tau_{K\setminus\{j\}},\eta_{-K})$ with positive probability under~$\sigma$. Indeed, $p|_{\theta_j}(\theta_{K\setminus\{j\}},\eta_{-K})>0$; for $k\in K\setminus\{j\}$, $\tau_k$ is in the support of $\sigma_k(\cdot|\theta_k)$; and for $k\notin K$, $\eta_k$ is in the support of $\sigma_k(\cdot|\eta_k)$, because~$\sigma$ is baseline-dominant and truth is weakly dominant. Thus $\pi_j(g(\tau_K,\eta_{-K})|\theta_j)=\pi_j(g(\theta_j,\tau_{K\setminus\{j\}},\eta_{-K})|\theta_j)$ by the first step, so~$\tau$ is not exposed by the information of~$\theta$.
\end{proof}

\begin{proof}[Proof of Theorem~\ref{Thm:necessity}]
Suppose that, for some $i\in N$, type~$\theta_i$ is not restrained at~$\theta_{-i}$. Then there is~$\tau_i$ with $g(\tau_i,\theta_{-i})\neq g(\theta)$ that is not exposed by the information of~$\theta_i$, i.e., $\pi_i(g(\tau_i,\eta_{-i})|\theta_i)\geq\pi_i(g(\theta_i,\eta_{-i})|\theta_i)$ for each $\eta_{-i}\in S_i(p|_{\theta_i})$. By \eqref{Eq:pointwise}, these payoffs are equal. We construct $\sigma\in\EE(M,\varphi,p)$ that does not induce~$g$ at~$\theta$.

For each~$j$ and~$\eta_j$, let $D_j(\eta_j)\subseteq M_j$ be the (non-empty) set of weakly dominant messages of type~$\eta_j$ in $(M,\varphi)$. Then $\varphi(d)=g(\eta)$ for each $\eta\in\Theta$ and each~$d$ with $d_j\in D_j(\eta_j)$ for each~$j$. Indeed, let $\sigma^q\in\EE(M,\varphi,q)$ induce~$g$ for the interior prior~$q$ in the statement. By Lemma~\ref{Lem:interior}, each $d_j\in D_j(\eta_j)$ is in the support of $\sigma^q_j(\cdot|\eta_j)$, so~$d$ is in the support of $\sigma^q(\cdot|\eta)$. Since $\eta\in\supp(q)$, $\varphi(d)=g(\eta)$.

Let every type of every agent other than~$i$, and every type of~$i$ other than~$\theta_i$, randomize uniformly on its set of weakly dominant messages, and let~$\theta_i$ randomize uniformly on the set~$B$ of its best replies to~$\sigma_{-i}$. Dominant messages are best replies to anything, so~$\sigma$ is a \BNE. Let $d_i\in D_i(\tau_i)$ and $e_i\in D_i(\theta_i)$. Against the dominant messages of types~$\eta_{-i}$, $d_i$ yields $g(\tau_i,\eta_{-i})$ and~$e_i$ yields $g(\theta_i,\eta_{-i})$. Hence
\begin{align*}
U_i(d_i,\sigma_{-i}|\theta_i)&=\sum_{\eta_{-i}\in S_i(p|_{\theta_i})}p|_{\theta_i}(\eta_{-i})\,\pi_i(g(\tau_i,\eta_{-i})|\theta_i)\\
&=\sum_{\eta_{-i}\in S_i(p|_{\theta_i})}p|_{\theta_i}(\eta_{-i})\,\pi_i(g(\theta_i,\eta_{-i})|\theta_i)=U_i(e_i,\sigma_{-i}|\theta_i).
\end{align*}
Since~$e_i$ is a best reply, so is~$d_i$, and $d_i\in B$. At~$\theta$, the profile consisting of~$d_i$ and dominant messages of the types~$\theta_{-i}$ has positive probability, and its outcome is $g(\tau_i,\theta_{-i})\neq g(\theta)$. So~$\sigma$ does not induce~$g$ at~$\theta$.

It remains to show that $\sigma\in\EE(M,\varphi,p)$. Let $\Delta>0$ be the smallest difference $U_i(e_i,\sigma_{-i}|\theta_i)-U_i(n_i,\sigma_{-i}|\theta_i)$ over $n_i\notin B$ (any $\Delta>0$ if $B=M_i$). For $\varepsilon\in(0,1)$, let each type $\eta_j\neq\theta_i$ play $\sigma^\varepsilon_j(\cdot|\eta_j):=(1-\varepsilon)\,\mathrm{unif}(D_j(\eta_j))+\varepsilon\,\mathrm{unif}(M_j)$, and, for $k\in\N$, let~$\theta_i$ play the logistic response with precision~$k$ to $\sigma^\varepsilon_{-i}$:
\[\sigma^{\varepsilon,k}_i(m_i|\theta_i):=\frac{e^{k\,U_i(m_i,\sigma^\varepsilon_{-i}|\theta_i)}}{\sum_{n_i\in M_i}e^{k\,U_i(n_i,\sigma^\varepsilon_{-i}|\theta_i)}}.\]
The strategy of agent~$i$ does not enter $\sigma^\varepsilon_{-i}$, so no fixed point is needed. Every type plays every message with positive probability, so each type with positive probability faces every message profile with positive probability.

The profile $\sigma^{\varepsilon,k}$ is weakly payoff-monotone. For~$\theta_i$, the logistic response is ordinally equivalent to expected payoffs. For another type~$\eta_j$, a message that is not weakly dominant is strictly worse than a dominant message against some message profile, and no better against any other; since~$\eta_j$ faces every message profile with positive probability, it is strictly worse in expectation. Dominant messages are payoff-equivalent. Since $\sigma^\varepsilon_j(\cdot|\eta_j)$ assigns a common probability to all dominant messages and a common and strictly smaller probability to all other messages, weak payoff-monotonicity holds for~$\eta_j$.

Expected payoffs are continuous in~$\varepsilon$. So for each~$k$ there is $\varepsilon_k\in(0,1/k)$ such that $|U_i(m_i,\sigma^{\varepsilon_k}_{-i}|\theta_i)-U_i(m_i,\sigma_{-i}|\theta_i)|<1/k^2$ for each $m_i\in M_i$. Let $\sigma^k:=\sigma^{\varepsilon_k,k}$. For $m_i,n_i\in B$, the ratio $\sigma^k_i(m_i|\theta_i)/\sigma^k_i(n_i|\theta_i)$ lies in $[e^{-2/k},e^{2/k}]$, which tends to~$1$. For $n_i\notin B$ and $m_i\in B$, the ratio $\sigma^k_i(n_i|\theta_i)/\sigma^k_i(m_i|\theta_i)$ is at most $e^{-k\Delta+2/k}$, which tends to~$0$. Hence $\sigma^k_i(\cdot|\theta_i)\to\mathrm{unif}(B)$, and every other type's strategy converges to uniform play on its dominant messages. So $\sigma^k\to\sigma$, and by \eqref{Eq:VB-equiv}, $\sigma\in\EE(M,\varphi,p)$.
\end{proof}

\begin{proof}[Proof of Remark~\ref{Lem:total-failure}]
Let $\sigma\in\EE(\Theta,g,p)$. Since~$\tau_i$ and~$\theta_i$ are weakly dominant for~$\theta_i$, they have positive frequency by Lemma~\ref{Lem:interior}, and each~$\theta_j$ reports truthfully with positive probability. So, at~$\theta$, the profile $(\tau_i,\theta_{-i})$ has positive probability, and its outcome differs from $g(\theta)$.
\end{proof}

\begin{proof}[Proof of Lemma~\ref{Lem:collective}]
(\ref{Lem:collective-suff}) Let~$\sigma$ be a baseline-dominant \BNE\ of $(\Theta,g,p)$ and let~$\tau$ be in the support of $\sigma(\cdot|\theta)$. By Lemma~\ref{Lem:unexposed}, $\tau$ is not exposed by the information of~$\theta$. If $g(\tau)\neq g(\theta)$, $\tau$ would be a consequential lie of~$\theta$ that is not exposed, contradicting that~$\theta$ is collectively restrained. So $g(\tau)=g(\theta)$. By Lemma~\ref{Lem:interior}, each \EE\ of $(\Theta,g,p)$ is a baseline-dominant \BNE.

(\ref{Lem:collective-nec}) Let~$\tau$ be a consequential lie of~$\theta$ that is not exposed by the information of~$\theta$, and let $L:=\{j\in N:\tau_j\neq\theta_j\}$, which is non-empty. For each~$k$ and each type~$\eta_k$, let $D_k(\eta_k)$ be the set of weakly dominant reports of~$\eta_k$ in $(\Theta,g)$. Since~$g$ is strategy-proof and weakly non-bossy, each report in $D_k(\eta_k)$ is outcome-equivalent to~$\eta_k$, i.e., $g(d_k,r_{-k})=g(\eta_k,r_{-k})$ for each $r_{-k}\in\Theta_{-k}$ (Sec.~\ref{Sec:model}). Let~$\sigma$ be the profile in which each type~$\theta_j$ with $j\in L$ randomizes uniformly on $D_j(\theta_j)\cup\{\tau_j\}$ and every other type with positive probability randomizes uniformly on its set of weakly dominant reports. Then~$\sigma$ is baseline-dominant, and at~$\theta$ it plays~$\tau$ with positive probability, so it does not induce~$g$ at~$\theta$. It remains to show that~$\sigma$ is a \BNE. Dominant reports are best replies, so only~$\tau_j$ for $j\in L$ needs checking. Let $j\in L$, and let~$r_{-j}$ be a report profile that~$\theta_j$ faces with positive probability under~$\sigma$, at a state $(\theta_j,\eta_{-j})$ with $\eta_{-j}\in S_j(p|_{\theta_j})$. Let $K:=\{j\}\cup\{k\in L\setminus\{j\}:\eta_k=\theta_k\textrm{ and }r_k=\tau_k\}$. Each other component of~$r_{-j}$ is a weakly dominant report of the corresponding type at~$\eta_{-j}$, hence outcome-equivalent to it. Replacing these components by the true types one at a time, $g(\tau_j,r_{-j})=g(\tau_K,\eta_{-K})$ and $g(\theta_j,r_{-j})=g(\theta_j,\tau_{K\setminus\{j\}},\eta_{-K})$, where $(\theta_{K\setminus\{j\}},\eta_{-K})=\eta_{-j}\in S_j(p|_{\theta_j})$. Since~$\tau$ is not exposed, $\pi_j(g(\tau_j,r_{-j})|\theta_j)\geq\pi_j(g(\theta_j,r_{-j})|\theta_j)$, and by \eqref{Eq:pointwise} the two payoffs are equal. Thus~$\tau_j$ and~$\theta_j$ yield the same expected payoff against~$\sigma_{-j}$, and since~$\theta_j$ is a best reply, so is~$\tau_j$.
\end{proof}

\begin{proof}[Proof of Lemma~\ref{Lem:remove}]
Let~$\tau$ be a report profile that is not exposed by the information of~$\theta$, and let $L:=\{j\in N:\tau_j\neq\theta_j\}$. We show that if $K\subseteq L$ and $j\in K\setminus\Ucal_p(\theta)$, then
\[g(\tau_K,\theta_{-K})=g(\tau_{K\setminus\{j\}},\theta_{-(K\setminus\{j\})}).\]
Undoing the misreports of the agents in $L\setminus\Ucal_p(\theta)$ one at a time then gives $g(\tau)=g(\tau_{\Ucal_p(\theta)},\theta_{-\Ucal_p(\theta)})$, since the agents in $\Ucal_p(\theta)\setminus L$ report truthfully.

Let $r_{-j}:=(\tau_{K\setminus\{j\}},\theta_{-K})$. Since $\theta\in\supp(p)$, $\theta_{-j}\in S_j(p|_{\theta_j})$, so Definition~\ref{Def:exposed} applies to agent~$j$, the set~$K$, and $\eta_{-K}:=\theta_{-K}$. Since~$\tau$ is not exposed, $\pi_j(g(\tau_j,r_{-j})|\theta_j)\geq\pi_j(g(\theta_j,r_{-j})|\theta_j)$, and by \eqref{Eq:pointwise} the two payoffs are equal. Since $j\notin\Ucal_p(\theta)$, $\theta_j$ is non-bossy in~$g$ or restrained in~$g$ by~$p$. If it is non-bossy, it is non-bossy at~$r_{-j}$, and the equality implies $g(\tau_j,r_{-j})=g(\theta_j,r_{-j})$. If it is restrained, it is restrained at~$r_{-j}$. Note that~$\tau_j$ is not exposed by the information of~$\theta_j$: this is the case $K=\{j\}$ of the exposure of~$\tau$. So~$\tau_j$ is not a consequential lie for~$\theta_j$ at~$r_{-j}$, and again $g(\tau_j,r_{-j})=g(\theta_j,r_{-j})$.
\end{proof}

\begin{proof}[Proof of Theorem~\ref{Thm:sufficiency}]
Let~$\tau$ be a report profile that is not exposed by the information of~$\theta$. If $\Ucal_p(\theta)=\emptyset$, then $g(\tau)=g(\theta)$ by Lemma~\ref{Lem:remove}. Let $\Ucal_p(\theta)=\{i\}$. By Lemma~\ref{Lem:remove}, $g(\tau)=g(\tau_i,\theta_{-i})$. Report~$\tau_i$ is not exposed by the information of~$\theta_i$: this is the case $K=\{i\}$ of the exposure of~$\tau$. Since~$\theta_i$ is restrained at~$\theta_{-i}$, $\tau_i$ is not a consequential lie for~$\theta_i$ at~$\theta_{-i}$, i.e., $g(\tau_i,\theta_{-i})=g(\theta)$. In both cases $g(\tau)=g(\theta)$, so each consequential lie of~$\theta$ is exposed, and~$\theta$ is collectively restrained.
\end{proof}

\begin{proof}[Proof of Lemma~\ref{Lm:full-implem-necessary}]
In both cases, $H$ contains the interior priors. Since \EE\ is non-empty, for each interior prior~$q$ some \EE\ of $(M,\varphi,q)$ induces~$g$ for~$q$, so Theorem~\ref{Thm:necessity} applies to $(M,\varphi)$.

Suppose, by contradiction, that~$g$ violates weak non-bossiness when~$H$ is the set of interior priors, or that it violates non-bossiness when~$H$ is the set of all priors. In both cases there are~$\theta$, an agent, say agent~$1$, and~$\tau_1$ such that $\pi_1(g(\tau_1,\theta_{-1})|\theta_1)=\pi_1(g(\theta)|\theta_1)$ and $g(\tau_1,\theta_{-1})\neq g(\theta)$. In the first case, by the definition of weak non-bossiness, these can be chosen so that $\pi_1(g(\tau_1,\eta_{-1})|\theta_1)=\pi_1(g(\theta_1,\eta_{-1})|\theta_1)$ for each $\eta_{-1}\in\Theta_{-1}$; let~$p$ be an interior prior. In the second case, let~$p$ place probability one on~$\theta$. In both cases, $p\in H$, and~$\tau_1$ is a consequential lie for~$\theta_1$ at~$\theta_{-1}$ that is not exposed by its information in~$p$. So~$\theta_1$ is not restrained at~$\theta_{-1}$. By Theorem~\ref{Thm:necessity}, some $\sigma\in\EE(M,\varphi,p)$ does not induce~$g$ at~$\theta$. This contradicts the hypothesis that each \EE\ of $(M,\varphi,p)$ induces~$g$ for~$p$.
\end{proof}

\begin{proof}[Proof of Theorem~\ref{Theorem-revelationpple}]
If~$g$ is strategy-proof, its direct-revelation mechanism is a dominant-strategy mechanism, so sufficiency in statements~1 and~2 follows from Corollaries~\ref{Th:str-pr} and~\ref{Thm:Interior}. Conversely, suppose that a dominant-strategy mechanism fully implements~$g$ in \EE\ for each interior prior. Then, for each interior prior, there is an \EE\ that induces~$g$, and~$g$ is strategy-proof by Lemma~\ref{Lm:full-interior-implementation-necessary}. Lemma~\ref{Lm:full-implem-necessary} gives non-bossiness in statement~1 and weak non-bossiness in statement~2.
\end{proof}

The following example shows that collective restraint is not necessary for every \EE\ to induce the SCF at a state (Remark~\ref{Rem:coalitions}).

\begin{example}\rm\label{Ex:dominance}
Let $N=\{1,2\}$, $\Theta_1=\{a,b,c\}$, $\Theta_2=\{a,b,z\}$, and $Z=\{0,\dots,5\}$. The SCF and the payoffs are
\[
\begin{array}{c|ccc}
g & a & b & z\\\hline
a & 0 & 0 & 3\\
b & 0 & 1 & 4\\
c & 0 & 2 & 5
\end{array}
\qquad
\begin{array}{r@{\;}l@{\qquad}r@{\;}l}
\pi_1(\cdot|a)=&(1,1,1,1,0,1/2), & \pi_2(\cdot|a)=&(1,1,0,0,0,0),\\
\pi_1(\cdot|b)=&(0,1,0,0,1,0), & \pi_2(\cdot|b)=&(0,1,1,0,0,0),\\
\pi_1(\cdot|c)=&(0,0,1,0,0,1), & \pi_2(\cdot|z)=&(0,0,0,1,1,1),
\end{array}
\]
where the~$k$-th entry is the payoff of outcome $k-1$. Truth is a best reply to each report of the other agent, so~$g$ is strategy-proof. Each lie that changes the outcome against some report is strictly worse than truth against some report, so~$g$ is weakly non-bossy. For instance, for type~$a$ of agent~1, lies~$b$ and~$c$ are strictly worse than truth against~$z$; for type~$a$ of agent~2, lie~$b$ is strictly worse against~$c$; and for type~$b$ of agent~2, lie~$z$ is costless against~$a$ but strictly worse against~$b$. Let information be complete at $\theta=(a,a)$. For both agents, $b$ is an unexposed lie, since $g(b,a)=g(a,b)=0$, and it is consequential against the other agent's report~$b$ at no cost to the agent: $\pi_1(1|a)=\pi_1(0|a)$ and $\pi_2(1|a)=\pi_2(0|a)$. So both types are bossy and not restrained. Each type is non-bossy, and thus restrained, at the other's type: agent~1 does not change the outcome against~$a$, and agent~2 changes it only with~$z$, which is costly. The joint misreport $(b,b)$ is not exposed, and, since $g(b,b)=1$, $\theta$ is not collectively restrained. Yet every \EE~$\sigma$ induces~$g$ at~$\theta$. For agent~1, $c$ weakly dominates~$b$, so $\sigma_1(c|a)\geq\sigma_1(b|a)$ (Lemma~\ref{Lem:interior}). If $\sigma_1(c|a)>0$, agent~2's reports~$b$ and~$z$ are costly, so agent~2 reports~$a$. Otherwise agent~1 reports~$a$. Either way the outcome is~$0$. The bad \NE\ $(b,b)$ is not an \EE\ because it requires agent~1 to play the weakly dominated~$b$ more often than~$c$.
\end{example}

\bibliography{ref-empirical-imp-no-url}

\end{document}